\documentclass[reqno]{amsart}
\usepackage[utf8]{inputenc}

\usepackage[backend=biber, giveninits=true, style=alphabetic, sorting=nyt,isbn=false, doi=false, maxalphanames=5,url=false, maxbibnames=99]{biblatex}
\usepackage[a4paper, total={6in, 8.5in}]{geometry}

\usepackage{amsmath}
\usepackage{amssymb}
\usepackage{amsthm}
\usepackage{bbm}
\usepackage{biblatex}
\usepackage[dvipsnames]{xcolor}
\usepackage{mathtools}
\definecolor{light-gray}{gray}{0.95}
\usepackage[color=light-gray]{todonotes}
\usepackage{tikz}
\usepackage[compat=1.1.0]{tikz-feynman}
\usepackage{tikz-cd}
\usepackage{rotating}
\usepackage{extarrows}
\usepackage{comment}
\usepackage[title]{appendix}
\usepackage[hidelinks]{hyperref}
\usepackage{thmtools}

\newtheorem{thm}{Theorem}[section]
\newtheorem{prop}[thm]{Proposition}
\newtheorem{lemma}[thm]{Lemma}
\newtheorem{cor}[thm]{Corollary}
\newtheorem{assumption}{Assumption}

\theoremstyle{definition}
\newtheorem{defn}[thm]{Definition}
\theoremstyle{remark}

\declaretheoremstyle[
  spaceabove=5pt, spacebelow=5pt,
  headfont=\itshape,
  notefont=\normalfont, notebraces={(}{)},
  bodyfont=\normalfont,
  postheadspace=1em,
  qed=$\Diamond$
]{pluto2}
    \declaretheorem[style=pluto2,name=Remark,    sibling=thm]{rmk}

\renewcommand{\epsilon}{\varepsilon}
\renewcommand{\phi}{\varphi}
\newcommand{\A}{\mathbb{A}}
\newcommand{\B}{\mathbb{B}}
\newcommand{\V}{\mathcal{V}}
\renewcommand{\d}{d^\nabla}

\newcommand{\R}{\mathbb{R}}
\newcommand{\N}{\mathbb{N}}
\newcommand{\Z}{\mathbb{Z}}

\newcommand{\D}{\mathcal{D}'}
\newcommand{\C}{C^\infty}
\newcommand{\FF}{\mathbb{F}}
\newcommand{\F}{\mathcal{F}}
\newcommand{\h}{\hbar}
\newcommand{\sfK}{\mathsf{K}}
\newcommand{\sfL}{\mathsf{L}}
\renewcommand{\L}{\mathcal{L}}

\newcommand{\tr}{\mathrm{tr}^\flat}
\renewcommand{\det}{\mathrm{det}^\flat}
\newcommand{\sdet}{\mathrm{sdet}^\flat}
\newcommand{\str}{\mathrm{str}^\flat}
\newcommand{\pair}[2]{\left\langle #1,\, #2 \right\rangle}

\DeclareMathOperator{\WF}{WF}
\DeclareMathOperator{\supp}{supp}

\DeclareMathOperator{\im}{im}

\title{Perturbative BF theory in axial, Anosov, gauge}

\author{Michele Schiavina}
\address{Department of Mathematics, University of Pavia, Via Ferrata 5, 27100 Pavia, Italy}
    \address{INFN Sezione di Pavia, via Bassi 6, 27100 Pavia, Italy}
    \email{michele.schiavina@unipv.it}

\author{Thomas Stucker}
\address{Department of Mathematics, ETH Zurich, R\"amistrasse 101, 8092, Z\"urich, Switzerland}
\email{thomas.stucker@math.ethz.ch}

\begin{document}

\begin{abstract}
The twisted Ruelle zeta function of a contact, Anosov vector field is shown to be equal, as a meromorphic function of the complex parameter $\hbar\in\mathbb{C}$ and up to a phase, to the partition function of an $\hbar$-linear quadratic perturbation of $BF$ theory, using an ``axial'' gauge fixing condition given by the Anosov vector field. Equivalently, it is also obtained as the expectation value of the same quadratic, $\hbar$-linear, perturbation, within a perturbative quantisation scheme for $BF$ theory, suitably generalised to work when propagators have distributional kernels. 
\end{abstract}

\maketitle

\tableofcontents

\section{Introduction}
Topological (quantum) field theories (TQFT's) are a valuable resource of insight for both pure and applied mathematics, for they are able to package relevant information, in both their classical and quantum formulation. Typically, the field theories ``know'' about the topology of the underlying manifold, e.g.\ through topological invariants as well as group representations and cohomology, and provide numerous deep links between seemingly unrelated areas of mathematics (see \cite{schwarz1,AtiyahTQFT,WittenJonesPoly,Berwick2015} just to name a few). 

An example of this paradigm is given by (Abelian) $BF$ theory \cite{schwarz1,BBRT91,cattaneo1995topological}, a model of great importance in its simplicity, for it serves as a nexus that links to several other scenarios. In certain cases, it can be seen as an example of Chern--Simons theory in three dimensions, thus connecting to knot invariants; it is a particular case of the Poisson sigma model in two dimensions, which in turn is related to Kontsevich's formality and deformation quantisation; and it can be seen as the ``topological backbone'' of theories such as Yang--Mills and even theories of gravity, which can be thought of as perturbations of $BF$ theory.

A first result that exemplifies the importance of TQFT's for pure mathematics, and the main motivator for this work, is derived from the quantisation of Abelian $BF$ theory:  The partition function of the theory (in the ``Lorenz gauge", see below) is the Analytic torsion of Ray and Singer \cite{ray_singer}, as shown by \cite{schwarz1,schwarz2} (see also \cite{cattaneo2017cellular,hadfield_kandel_schiavina} for two alternative approaches).

To properly address the quantisation of Abelian $BF$ theory, and rigorously phrase the previous claim, one needs to address the fact that it is a \emph{gauge} theory, i.e.\ it is invariant under the action of an infinite dimensional group of local symmetries. This degeneracy is nonphysical, and it needs to be ``removed'' in the appropriate way in order to extract meaningful observable (i.e.\ physical) information from the theory. This procedure is called \emph{gauge fixing}, and it is conveniently addressed within a rigorous approach to quantisation of the theory by means of the Batalin--Vilkovisky formalism \cite{BV81,BV83}. 

A \emph{gauge fixing} can be thought of as an arbitrary choice to present the quotient of the space of fields by the group action within the space of fields itself,\footnote{In good cases this is the image of a section of the quotient map, although this is generally too much to ask for.} and it is expected to be an immaterial choice on physical grounds. Note that all mentioned results that link Abelian $BF$ theory to the analytic torsion are obtained by choosing a particular gauge fixing condition that goes under the name of ``Lorenz gauge,'' which depends on the choice of an arbitrary metric on the underlying manifold.

Recently, one of the authors noted how a different choice of gauge fixing condition, which instead depends on the choice of an Anosov-contact structure on an odd-dimensional manifold, leads to the statement that the partition function for Abelian BF theory returns the (absolute value of the) value at zero\footnote{To be precise, its absolute value or the reciprocal thereof, according to the dimension.} of a meromorphic function called (twisted) Ruelle zeta \cite{ruelle}, when it is computed using this gauge fixing condition \cite{hadfield_kandel_schiavina}. This result is immediately tied to a conjecture due to Fried, stating that the value at zero of the (twisted) Ruelle zeta function over a manifold that admits an Anosov vector field should compute exactly the analytic torsion of the manifold \cite{fried,fried1987lefschetz}. The conjecture is true if the quantisation procedure can be shown not to depend on the arbitrary choice of gauge fixing.\footnote{Note however that we can reasonably predict physical quantities to be only \emph{locally constant} on smooth families of gauge fixing conditions, so there is an overarching question related to the existence of a nontrivial moduli space of gauge fixing conditions, even in the best case scenario.}

While in this paper we will not discuss the invariance of the theory on the choice of gauge fixing (which is done in the companion paper \cite{Schiavina_Stucker1}), we will complete the result of \cite{hadfield_kandel_schiavina} by showing that the \emph{whole} Ruelle zeta function, i.e.\ as a meromorphic function on the complex plane, can be given a rigorous field theoretic presentation, either by means of the expectation value of a suitably chosen \emph{quadratic} functional of the fields, or as the partition function of a perturbation of $BF$ theory by means of the same functional. This, together with the result in \cite{hadfield_kandel_schiavina}, completely reconstructs the Ruelle zeta function perturbatively, by means of Feynman diagrams.

In order to do this, despite $BF$ theory being ``simple'' from the perturbative point of view, a significant adaptation of the standard methods used to compute Feynman integrals (e.g.\ those presented in \cite{costello}) becomes necessary. This is due to the simple, but crucial, reason that the operators that appear naturally after imposing the contact, Anosov, gauge fixing condition are not elliptic, and thus their heat kernels are not smooth. 
As a matter of fact, $BF$ theory in this gauge features the operator $\L_X$ acting on differential forms, where $X$ is a contact, Anosov vector field, and perturbative quantisation of the theory only makes sense via a careful manipulation of the wavefront sets of the propagator of the theory. This requires phrasing perturbative quantisation within the setting of microlocal analysis. Hence, besides discussing the particular example of BF theory, we employ a generalisation of the perturbative quantisation techniques to a much larger class of operators.

The paper is structured as follows. In Section \ref{section_perturbation_theory} we provide an introduction to perturbative quantisation of gauge field theories in the BV formalism and set the main definitions and methods we will use in the remainder of the paper.

In section \ref{sec:RuelleZetaFT} we define the geometric setting we will need to define BF theory (Section \ref{BF_theory}) and the twisted Ruelle zeta function (Section \ref{sec:Ruellezeta}). 

The rest of the paper is devoted to recovering the twisted Ruelle zeta function from the quantisation of BF theory. 
To this end, we introduce a functional with which we perturb the action of $BF$ theory at first order in $\hbar$. 

We first show (Section \ref{sec:perturbingfunctional}) that, up to a phase, the full Ruelle zeta function (as a function of $\hbar\in\mathbb{C}$) is the partition function of the perturbed action in the contact gauge (Theorem \ref{thm:Result_PF}). This requires a regularisation scheme for determinants of operators on infinite dimensional vector spaces that is a generalisation of the (more widespread) zeta-funztion regularisation. We call it flat-regularisation, since it uses the flat trace and flat determinant tools.

Then, we change perspective slightly and consider the perturbing functional as an observable, of which we compute the expectation value with respect to free, Abelian, $BF$ theory, in the contact gauge (Section \ref{section_diagrams_ruelle}). This time we obtain that the ratio of Ruelle zeta function by its value at zero can be obtained as the perturbative expectation value of our perturbing functional, thus effectively reconstructing the whole zeta function (Theorem \ref{thm:result_expectation}).

\section*{Acknowledgements}
This work stems from the Master thesis of T. Stucker at ETH Zurich. M.S.\ acknowledges partial support of SwissMAP, since part of this work has been developed while employed by ETH Zurich. We are both grateful to G.\ Felder for helpful discussions and encouragement.

\section{Perturbation Theory and Effective Field Theory}
\label{section_perturbation_theory}

In this section we introduce the relevant notions of classical field theory with local symmetries and discuss its perturbative quantisation in the effective field theory sense. Our framework of choice to handle perturbative quantisation of gauge theories, is the Batalin--Vilkovisky (BV) cohomological framework. We will outline the basics of perturbative quantisation for gauge teories in the BV formalism, starting from a finite-dimensional scenario.

\medskip

\subsection{Finite-dimensional perturbation theory}\label{section_perturbation_theory_details}

Let $V$ be a finite dimensional vector space. The stationary phase formula provides a small $\hbar$ asymptotic expansion of the oscillatory integral
$$\int_V e^{\frac{i}{\h}f(x)}\,dx.$$
We apply this perturbative expansion to a function
$$f(x) = \frac{1}{2}\mathbf{Q}(x,x) + I(x),$$
where $\mathbf{Q}(x,x)$ is a non-degenerate quadratic form on $V$, the free action, and $I(x)$ is a polynomial, called the interaction term, and thought of as a perturbation of the free action. One way to write the stationary phase formula in this case (see, e.g., \cite{mnev}) is
\begin{equation}
\label{stationary_phase}
    \int_V e^{\frac{i}{\h}(\frac{1}{2}\mathbf{Q}(x,x)+I(x))}\,dx \sim \int_V e^{\frac{i}{2\hbar}\mathbf{Q}(x,x)}\,dx \cdot e^{\frac{1}{\h}\Gamma(\mathbf{Q}^{-1},\,I)(0)}\,, \quad \mathrm{as}\quad \h \to 0,
\end{equation}
where the right hand side is interpreted as a formal power series in the parameter $\h$. The term $\Gamma(\mathbf{Q}^{-1},\,I)$ is defined as a sum over connected Feynman diagrams:
\begin{equation}
\label{feynman_expansion}
    \Gamma(\mathbf{Q}^{-1},\,I)(0) = \sum_\gamma \frac{\h^{l(\gamma)}}{|\mathrm{Aut}(\gamma)|}\Phi_\gamma(i\mathbf{Q}^{-1},\{iI_d\})(0),
\end{equation}
where a Feynman diagram $\gamma$ is a graph specified by the following data:
\begin{itemize}
    \item a finite set $V(\gamma)$ of vertices,
    \item a finite set $H(\gamma)$ of half-edges,
    \item a map $i: H(\gamma) \to V(\gamma)$, assigning each half-edge to the vertex, to which it is attached,
    \item an involution $\sigma: H(\gamma) \to H(\gamma)$, whose set of fixed points, $T(\gamma)$, is the set of tails of the diagram, and whose set of two-element orbits, $E(\gamma)$, is the set of edges.
\end{itemize}

A Feynman diagram is \emph{connected} if its underlying graph is connected, and the sum in \eqref{feynman_expansion} is over connected diagrams $\gamma$. Each term is multiplied by a power of $\h$, where the exponent $l(\gamma)$ is the number of loops occurring in the graph $\gamma$. The other prefactor, $|\mathrm{Aut}(\gamma)|$, called the symmetry factor of $\gamma$, is the cardinality of the automorphism group of the graph.

To each Feynman diagram $\gamma$ we associate its weight $\Phi_\gamma(i\mathbf{Q}^{-1},\{iI_d\})$, which is defined as a function on the vector space $V$, depending on the quadratic form $Q$ and the interaction term $I$ of the action functional. 

Let us write $I(x)=\sum_d I_d(x)$, where $I_d$ is a homogeneous polynomial of degree $d$, viewed as elements of the $d$-th tensor power of the dual space: $I_d \in (V^*)^{\otimes d}$. Then, denote by $\mathbf{Q}^{-1}$ the inverse of (the matrix representing) the quadratic form $\mathbf{Q}$. Note that, in finite dimensions, the inverse of $\mathbf{Q}$ can canonically be viewed as an element $\mathbf{Q}^{-1} \in V\otimes V$. This is the \emph{propagator} of the theory.

Given a Feynman diagram $\gamma$ and a vector $v\in V$, the \emph{weight of (the Feynman diagram) $\gamma$} is constructed as follows:
\begin{enumerate}
    \item To each tail in $T(\gamma)$ we associate a factor of $v$ and to each edge in $E(\gamma)$ we associate a factor of $i\mathbf{Q}^{-1}$. Taking the tensor product of the factors associated to tails and edges returns an element of $V^{\otimes |H(\gamma)|}$.
    \item To each vertex of order $d$, i.e.\ a vertex with $d$ half-edges attached to it, we assign a factor of $iI_d$. Taking the tensor product similarly produces an element of $(V^*)^{\otimes |H(\gamma)|}$.
    \item The weight $\Phi_\gamma(i\mathbf{Q}^{-1},\{iI_d\})(v)$ is the contraction of these two elements of the respective tensor power. 
\end{enumerate}

\begin{rmk}
The integral in \eqref{stationary_phase} is not necessarily convergent, even for finite-dimensional $V$. In order for the stationary phase formula to be an actual asymptotic expansion, the integral should be taken with respect to a compactly supported measure. {Although there are ways to make rigorous sense of the formula \eqref{stationary_phase}, we ignore these subtleties here}, since in infinite dimensions we just take Equation \eqref{stationary_phase} as a definition.
\end{rmk}

\begin{rmk}
The factors of the imaginary unit $i$ are introduced since we are considering an oscillatory integral.\footnote{There is a similar formulation for Gaussian integrals without the $i$'s.} The vector $v$ can be viewed as an ``external field'', and in the expression \eqref{stationary_phase} we are setting this external field to zero, which is equivalent to only considering Feynman diagrams with no tails. However, we will need the formulation with the external field shortly, when we consider effective interactions.
\end{rmk}

The idea of perturbation theory is to use the perturbative expansion to assign a value to oscillatory integrals, even when the vector space $V$ is infinite-dimensional. Going back to Equation \eqref{stationary_phase}: from the perturbation theory perspective, one needs to normalize expectation values of interaction polynomials using the free-theory (given by the quadratic form $\mathbf{Q}$) as a benchmark.

In other words, one is interested in the quantity
\begin{equation}
\label{expectation_value_finite}
    \biggl\langle e^{\frac{i}{\hbar}I} \biggr\rangle := \frac{\int_V e^{\frac{i}{\h}(\frac{1}{2}\mathbf{Q}(x,x)+I(x))}\,dx}{\int_V e^{\frac{i}{2\h}\mathbf{Q}(x,x)}\,dx} = e^{\frac{1}{\h}\Gamma(\mathbf{Q}^{-1},\,I)(0)}\,,
\end{equation}
interpreted as the expectation value of $\exp(\frac{i}{\hbar}I)$ with respect to the free theory. Note that this can also be interpreted as the interacting partition function, with the free partition function set to one. Observe furthermore that the right hand side is now expressed solely in terms of the Feynman diagram expansion, which one can hope to make sense of even in infinite dimensions.

\begin{rmk}
In the case where the action functional is quadratic, which is often referred to as a \emph{free theory}, a simplification of the perturbative approach is available: The partition function for such a quadratic functional can formally be interpreted as an oscillatory, Gaussian, integral. In finite dimensions the result of this integral can be expressed in terms of the determinant of the matrix in the quadratic form (up to a phase)
\[
Z=\int_V e^{\frac{i}{2\hbar}\mathbf{Q}(x,x)}dx \propto \mathrm{det}(\mathbf{Q})^{-\frac12}
\]
Given the appropriate generalisation of the determinant to infinite-dimensional operators, which requires a choice of regularisation, one can \emph{define} the partition function of a quadratic functional by extending the oscillatory, Gaussian, integral formula to the infinite-dimensional case. In particular, when we have access to the value of the free partition function one can use the perturbative framework to compute the interacting partition function.  Cf. Remark \ref{phase_factor}.
\end{rmk}

In Appendix \ref{section_flat_det} we will introduce an appropriate notion of determinant for operators on function spaces, called \emph{flat determinant}. Foreshadowing what will come, we phrase the following definition in terms of $\det$. 

\begin{defn}
\label{fresnel}
Let $S: \V \to \R$ be a functional of the form $S(v) = \pair{v}{Av}$, where $\V$ is a (possibly infinite-dimensional) graded vector space endowed with an inner product $\pair{\cdot}{\cdot}$ and $A:\V \to \V$ is a degree-preserving linear operator with well-defined flat determinant. Then we define the value of the oscillatory, Gaussian, ``integral" as
\begin{equation}
    \int_V e^{iS} \doteq  |\sdet(A)|^{-\frac{1}{2}}= \bigl|\prod_k \det(A_k)^{(-1)^k}\bigr|^{-\frac{1}{2}},
\end{equation}
where $A_k=A\vert_{\mathcal{V}^{(k)}}$ acts on homogeneous elements of degree $k$.
\end{defn}
\begin{rmk}
\label{phase_factor}
The integral over summands of odd degree in $\mathcal{V}$ results in a determinant raised to the opposite sign of what is expected from the usual oscillatory, Gaussian, integral, see e.g.\ \cite{mnev}. Compared to an $n$-dimensional oscillatory, Gaussian, integral, we have dropped the constant factor of $(2\pi)^{\frac{n}{2}}$ and also ignored the phase factor $\exp(i\frac{\pi}{4}\mathrm{sign}(A))$, which will not matter for the applications considered in the present paper.
\end{rmk}

\subsection{Effective Field Theory}
\label{EFT}

A local, Lagrangian, field theory is the data of a \emph{space of fields} $\F=\C(M,F)$, modeled on sections of vector bundles,\footnote{We will assume that $M$ is a closed manifold, which ensures $\mathcal{F}$ is equipped with a nuclear Fr\'echet smooth structure. See e.g.\ \cite[Section 30.1, unnumbered Proposition]{kriegl_michor} or \cite[Appendix 2]{costello}.} and a \emph{local action functional} $S: \F \to \R$, of the form\footnote{The infinite jet evaluation is a map $j^\infty\colon\mathcal{F}\times M\to J^\infty F$, such that $j^\infty(\phi,x) = (j^\infty\phi)(x)$, with $j^\infty\phi$ the infinite jet equivalence class of $\phi\in \mathcal{F}$.  See e.g.\  \cite{Anderson:1989,AndersonTorre,BlohmannLFT}.}
\begin{equation}
\label{local_functional}
    S(\phi) = \int_M ((j^\infty)^*\sfL)(\phi)(x),
\end{equation}
where $\sfL\in\Omega^{0,\mathrm{top}}(J^\infty F)$---often called a Lagrangian density---only depends on a finite-order jet of a section $\phi \in \F$.

Since the space of fields is infinite dimensional, a measure-theoretic interpretation of the typical expressions used in quantum theory, like:
\begin{equation}
\label{partition_function}
    Z = ``\int_\F e^{\frac{i}{\hbar}S(\phi)}\,\mathcal{D}\phi \,", \qquad \langle\mathcal{O}\rangle = `` \int_\F \mathcal{O}(\phi) e^{\frac{i}{\hbar}S(\phi)}\mathcal{D}\phi\," ,
\end{equation}
is generally not available.
(Here $\mathcal{O}$ is a functional on $\F$, the parameter $\h$ is formal, and $\mathcal{D}\phi$ is intended to be a measure on $\F$.) However, one can make sense of these expressions perturbatively, i.e.\ as formal power series in $\hbar$, as described in detail in Section \ref{section_perturbation_theory_details}.

While for quadratic functionals, like the free part of a general action, we could generalise the finite dimensional scenario, there still remain problems in evaluating Feynman diagrams in infinite dimensions. We illustrate this with a simple example. Assume that our space of fields is $\F\doteq\C(M)$, the space of smooth functions on some compact manifold $M$, and we are given a free theory defined by the functional
$$S_0(\phi) = \frac{1}{2}\int_M \phi(x) D \phi(x)\,dx,$$
where $D$ is an injective differential operator on $M$. {(In terms of the picture outlined above the quadratic functional $S_0$ corresponds to $\frac12\mathbf{Q}(x,x)$.)} The propagator of the theory is then the Schwartz kernel of $D^{-1}$, which in general will only be a distribution on $M\times M$, denoted\footnote{Here $\mathcal{D}'(N)$ denotes distributional sections of the trivial $\mathbb{R}$-bundle over a manifold $N$.} $P \in \D(M\times M)$. 

Given an interaction term $I(\phi)$, homogeneous of degree $d$, we can view it as an element $I\in\D(M^d)$ and attempt to compute the Feynman diagram expansion $\Gamma(P,I)(\phi)$ for some external field $\phi$. However, we run into problems when trying to apply the Feynman rules described above. Since both the propagator $P$ and the interaction $I$ are distributions, there is, in general, no way to contract them, and thus construct the Feynmann weight $\Phi_\gamma(i\mathbf{Q}^{-1},\{iI_d\})(\varphi)$. 

One way to deal with this problem is through the framework of effective field theory. Here, we follow Costello's formulation of effective quantum field theory \cite{costello}.\footnote{Our exposition differs slightly, since we have a factor of $i$ in the partition function, whereas Costello considers $\exp(\frac{1}{\h}S)$ as the ``integrand" in the partition function.} The idea is to replace the interaction term $I$ with a family of \emph{effective interactions}\footnote{Note that, although it is possible to make sense of smooth functions on the nuclear Fréchet space $\F=\C(M)$, our main example will be that of a bilinear but nonlocal, operator $\F\times \F\to \mathbb{R}$.} 
$$\{I[L]\in\C(\F), L\in\R_+\},$$
with $L$ determining the \emph{length scale}.  The offending propagator $P$ is replaced by its heat kernel regularised counterpart, that is, to $L_1, L_2 \in \R_+$ (interpreted as two length scales) we associate the propagator:
$$P_{L_1,L_2} = \int_{L_1}^{L_2}K_t\,dt,$$
where $K_t$ is the heat kernel for $D$, i.e.\ the Schwartz kernel of the operator $\exp(-tD)$. When $D$ is an elliptic operator with strictly positive spectrum, this propagator is a smooth function, i.e.\ $P_{L_1,L_2} \in \C(M\times M)$ for all $L_1,L_2 > 0$, where we can even set $L_2=\infty$. Thus, $P_{L_1,L_2}$ can be contracted with distributions and the Feynman rules are well-defined for this propagator. 

In order to link together the behaviours at different values of the parameter $L$, the effective interactions $I[L]$ are required to satisfy the renormalisation group equation:
\begin{equation}
\label{RGE}
iI[L_2](\phi) = \Gamma(P_{L_1,L_2},\, I[L_1])(\phi),
\end{equation}
relating the scale-$L_2$ to the scale-$L_1$ interactions. Then, the expectation value in \eqref{expectation_value_finite} can be computed as 
\begin{equation}
\label{expectation_value}
    \biggl\langle e^{\frac{i}{\hbar}I} \biggr\rangle = e^{\frac{1}{\h}\Gamma(P_{L,\infty}, I[L])(0)},
\end{equation}
and the renormalisation group equation implies that this is independent of the length scale-$L$ used in the expression. Once again, this expectation value can also be interpreted as the partition function of the interacting theory, where the undefined free partition function has been set to $1$.

\begin{rmk}
Starting from a local interaction $I$ one can obtain a family of effective interactions by renormalisation. The scale-$L$ interaction is given by 
\begin{equation}
\label{renormalization}
    iI[L](\phi) = \lim_{\epsilon\to 0}\Gamma(P_{\epsilon,L},I)(\phi),
\end{equation}
provided the limit exists. When this is not the case, one subtracts counterterms $I^{CT}(\epsilon)$ from the interaction to make the limit of the renormalised interaction well-defined:
$$iI^R[L](\phi) = \lim_{\epsilon\to 0}\Gamma(P_{\epsilon,L},I-I^{CT}(\epsilon))(\phi).$$
We will not make use of this technique, since the example we will study in Section \ref{section_ruelle_zeta_perturbative} does not require any counterterms for the limit in \eqref{renormalization} to exist.
\end{rmk}

\subsection{Effective Field Theory in the BV Formalism}
\label{EFT_BV}

Perturbative, effective, field theory needs to be amended when dealing with \emph{gauge theories}. These are models that admit local symmetries, i.e.\ a set of (local) vector fields preserving the action functional (possibly up to total derivatives). In this case, the field theoretic description is characterised by redundancies that must be removed, in order to extract sensible information from the model. The removal of such symmetry is called \emph{gauge-fixing}.

In the presence of gauge symmetries, the quadratic form $\mathbf{Q}$ becomes degenerate \cite{henneaux_teitelboim} and the asymptotic expansion \eqref{feynman_expansion} no longer makes sense. As a consequence, for the perturbative methods to be well-defined, one needs to implement gauge fixing in perturbative effective field theory. 

Indeed, the construction outlined above can be combined with the Batalin--Vilkovisky (BV) formalism for gauge field theories \cite{BV81,BV83}. Here, we provide a brief summary of the construction of effective theories within the BV formalism, following \cite{costello2,costello}, and \cite{mnev} (see also \cite{ZinnJustin}).

The BV formalism deals with gauge theories by expanding the original space of fields $\F$ to a graded $(-1)$-symplectic manifold $\F_{BV}$, and replacing the original action functional with a functional $S_{BV}$ on $\F_{BV}$. More precisely:

\begin{defn}[Classical BV Theory]
\label{def:BV_classical}
A classical BV theory is given by the data $(\F_{BV}, \Omega_{BV}, S_{BV}, Q_{BV})$, where
\begin{itemize}
    \item the \emph{space of BV fields} $\F_{BV}$ is a $\Z$-graded manifold,
    \item the \emph{BV} form $\Omega_{BV}$ is a degree $-1$, local, (weak) symplectic form on $\F_{BV}$,
    \item the \emph{classical BV operator} $Q_{BV}$ is a degree $1$, local, vector field on $\F_{BV}$ with the cohomological property $[Q_{BV},Q_{BV}]=0$, 
    \item the \emph{BV action} $S_{BV}$ is a local function on $\F_{BV}$ of degree $0$, satisfying the Hamiltonian condition and the classical master equation: 
    \[
    \iota_{Q_{BV}}\Omega_{BV} = \delta S_{BV}, \qquad \{S_{BV},S_{BV}\}_{\Omega_{BV}} = 0.
    \]
\end{itemize}
Here, $\{\,\cdot\,,\,\cdot\,\}_{\Omega_{BV}}$ is the Poisson bracket on (Hamiltonian) functions induced by the (weak) symplectic structure.

\end{defn}

In our applications, $\F_{BV}$ is the space of sections of some graded vector bundle\footnote{To connect with the classical theory we want $\FF^{(0)}=F$.} $\FF\to M$ over a compact, connected, orientable manifold $M$ without boundary, i.e.\ $\F_{BV} = \C(M,\FF)$.
Since $\F_{BV}$ is a vector space, the graded weak-symplectic structure can be defined directly as a graded anti-symmetric, bilinear, pairing on the space of fields:
$$\Omega_{BV}: \F_{BV}\times\F_{BV} \to \R.$$
The pairing then induces a map $\Omega_{BV}^\flat: \F_{BV} \to \F_{BV}^*$ into the (strong) dual space to the space of fields, and we require this map to be injective. (See \cite[Section 48]{kriegl_michor} for a treatment of infinite dimensional symplectic manifolds.)

In this setting, the notion of Lagrangian subspace requires an appropriate generalisation. We employ the following (cf.\ \cite{weinstein,cattaneo-contreras-split}):

\begin{defn}
\label{def:lagrangian_subspace}
Let $(\V,\Omega)$ be a (weak) symplectic vector space. A subspace $\sfL\subset \V$ is said to be Lagrangian iff it is isotropic, i.e.\ $\Omega\vert_{\sfL}=0$ and there is a splitting $\V = \sfL\oplus\sfK$, where $\sfK$ is also isotropic.
\end{defn}

We write the quadratic, free, part of the (BV) action as 
$$S_0(\phi) = \frac{1}{2}\Omega_{BV}(\phi,Q\phi), \quad \phi\in\F_{BV},$$
where $Q$ is a differential operator on $\F_{BV}$, of degree $1$ with respect to the grading, which squares to zero. Note that $Q^2 = 0$ is equivalent to the classical master equation for $S_0$. 
\begin{rmk}[$Q$ vs. $Q_{BV}$]\label{rmk:QvsQBV}
    Note that $Q$ and $Q_{BV}$ may differ. In full generality, $Q_{BV}$ may contain nonlinear terms, owing to the fact that it is the Hamiltonian vector field of the possibly non-quadratic action $S_{BV}$. A typical example is given by $Q=d$ the de Rham differential, while $Q_{BV}$ may act as $d_A$ for some connection $A$, which could be a field configuration. In the perturbative quantisation setting it may be convenient (but not strictly necessary) to work with $Q$ and treat nonlinear contributions as part of the interaction. In our main application, Abelian BF theory, there is no distinction between $Q$ and $Q_{BV}$.
\end{rmk}

To apply the perturbative framework to a gauge theory in the BV Formalism, as mentioned, we need a gauge fixing condition. 

\begin{defn}\label{def:GFoperator}
A gauge fixing operator is a linear, degree $-1$, local map $Q_{GF}\colon \F_{BV} \to \F_{BV}$ such that $Q_{GF}^2=0$,  and such that $\sfL\doteq \mathrm{im}(Q_{GF})$ is a Lagrangian subspace. Moreover, we require that  
    \[
    D=[Q_{GF},Q]= QQ_{GF} + Q_{GF}Q
    \] 
be a differential operator, and that the Lagrangian splitting $\F=\sfL \oplus \sfK$ be $D$-invariant.
\end{defn}

\begin{rmk}[Comparison with \cite{costello}]
\label{remark_costello_comparison}
In the framework outlined \emph{ibidem}, a gauge fixing operator is instead required to induce a splitting of the form
\begin{equation}
\label{costello_splitting}
    \F_{BV} = \im(Q)\oplus\im(Q_{GF})\oplus\ker(D).
\end{equation}
We will only deal with the case $\ker(D)=\{0\}$, although we will not be able to achieve this splitting in the space of \emph{smooth} sections for the particular gauge-fixing condition we will be interested in (see Remark \ref{rmk:costellospilittinganisotropic}).  
The idea of extracting the kernel of $D$ from $\F_{BV}$ and choosing a Lagrangian splitting of the complement is explored in \cite{cattaneo-mnev-reshetikhin-quantum} as well, where it relates to a choice of ``residual fields,'' on which the resulting effective field theory will depend. We will not pursue this idea further, as it lies beyond the scope of this work. We note here that both $Q$ and $Q_{GF}$ commute with $D$.

Another restrictive condition that is imposed on the free action and the gauge fixing operator by the framework of \cite{costello} is the requirement that $D$ be a generalised Laplacian. Among other things, this guarantees the existence of a smooth heat kernel. We shall shortly dispose of this requirement and work with some alternative regularisation procedures to make sense of the Feynman rules. 
\end{rmk}

When working with a quadratic BV theory, for which we can write $S_{BV} = \frac12 \Omega_{BV}(\varphi, Q_{BV}\varphi)$, we can look at the restriction $S_{BV}\vert_\sfL$ of the quadratic functional onto the gauge-fixing subspace. 
One can then consider the gauge-fixed partition function, as the ``would-be oscillatory, Gaussian, integral'' for the gauge fixed action-functional. Notice that restricting to $\sfL$ involves the Jacobian of the gauge-fixing condition (see, e.g., \cite[Remark 31]{hadfield_kandel_schiavina}). We will work with the following generalisation of the Gaussian-integral interpretation of the partition function:
\begin{defn}\label{def:partitionfunction}
Let $(\F_{BV}, \Omega_{BV}, S_{BV}, Q_{BV})$ be a BV theory with quadratic action functional $S_{BV}=\Omega_{BV}(\varphi,Q_{BV}\varphi)$ and with gauge fixing operator $Q_{GF}$, so that $D=[Q_{GF},Q_{BV}]$ and denote $\im(Q_{GF})=\sfL$. The gauge-fixed partition function is:\footnote{To compute $|\sdet(Q_{GF}\vert_{\sfK})|$ it is convenient to introduce an auxiliary metric, w.r.t. which one can look at the dual operator $Q_{GF}^*$ and compute $|\sdet(Q_{GF}Q_{GF}^*)|^{\frac{1}{2}}$.}
\[
Z(S_{BV},\sfL) = |\sdet(Q_{GF}|_\sfK)|^{\frac{1}{2}} \cdot |\sdet(D\vert_{\sfL})|^{-\frac{1}{2}}.
\]
This approach generalises to BV actions that are not quadratic, for which the decomposition $S_{BV} = S_0 + I$ holds, with $S_0= \frac12\Omega_{BV}(\varphi, Q \varphi)$ as above. The free partition function of the theory is 
\[
Z^{\text{free}}(S_{BV},\sfL) \doteq Z(S_0,\sfL).
\]
\end{defn}

We will define the kernel of an operator using the symplectic pairing $\Omega_{BV}$. To this end, note that the pairing $\Omega_{BV}$ induces a map 
$$1\otimes\Omega_{BV}: \F_{BV}\otimes\F_{BV}\otimes\F_{BV} \to \F_{BV}.$$ 
For $\F_{BV} = \C(M,\FF)$ this can be extended to a map:
$$1\otimes\Omega_{BV}: \D(M\times M, \FF\boxtimes \FF)\otimes\C(M,\FF) \to \D(M,\FF).$$
Given an operator $B$ on $\F_{BV}$, we can then define its kernel $K_B\in\D(M\times M, \FF\boxtimes\FF)$ by
\begin{equation}
\label{BV_kernel}
    K_B * \phi := (-1)^{|B|}(1\otimes\Omega_{BV})(K_B\otimes\phi) = B\phi, \quad \forall\,\phi\in\F_{BV},
\end{equation}
where $|B|$ is the degree of the operator with respect to the grading.
Note that $K_B$ is akin to the Schwartz kernel of $B$, but instead of using integration with respect to a density on $M$ to define the pairing with functions, we use the symplectic structure $\Omega_{BV}$. In this sense, it can be thought of as a generalisation of the Schwarz kernel for symmetric pairings to antisymmetric bilinear forms. See \cite{costello} for further details. In the context of BF theory, introduced in Section \ref{BF_theory}, this definition of the kernel through the symplectic pairing should be compared to Remark \ref{remark_schwartz_kernel}.

\begin{defn}\label{def:propagator}
The heat kernel of $D=[Q_{GF},Q]$ is the distribution $K_t$, such that
\begin{equation}
\label{heat_kernel_costello}
    K_t * \phi = (1\otimes\Omega_{BV})(K_t\otimes\phi) = e^{-tD}\phi.
\end{equation}
The propagator from scale $L_1$ to $L_2$ is
\begin{equation}
\label{propagator_costello}
    P_{L_1,L_2} \doteq \int_{L_1}^{L_2}(Q_{GF}\otimes1)K_t\,dt,
\end{equation}
so that
$$P_{L_1,L_2} * \phi = \int_{L_1}^{L_2} Q_{GF}\,e^{-tD}\phi\,dt = \int_{L_1}^{L_2} e^{-tD}Q_{GF}\phi\,dt.$$
\end{defn}

\begin{rmk}
We can think of $P_{0,\infty}$, as the kernel of the inverse of $Q$ restricted to the gauge fixing subspace, since for $\phi\in\im(Q_{GF})$ we have
$$
P_{0,\infty} * Q\varphi = \int_0^\infty e^{-tD}Q_{GF}Q\phi\,dt = \int_0^\infty e^{-tD}D\phi\,dt = -\int_0^\infty \frac{d}{dt}\left(e^{-tD}\phi\right)\,dt = \phi,$$
assuming that $\exp(-tD)\phi$ vanishes as $t\to\infty$.
\end{rmk}

Using this propagator, we can introduce effective interactions for the BV theory. As above, these are families of functionals $I[L] \in \C(\F_{BV})$, satisfying the renormalisation group equation \eqref{RGE}. (Note that the renormalisation group equation now depends on the choice of gauge fixing through the $Q_{GF}$ entering the propagator, as does the expectation value defined in \eqref{expectation_value}.) 

The family of functionals $\{I[L],L\in\mathbb{R}_+\}$ must satisfy certain consistency conditions in order for the perturbative field theory to make sense. Since these are not going to be relevant to this work, we refer to the definition of an effective BV theory in \cite{costello}. (A version of these consistency conditions adapted to the case at hand is given in \cite{ThomasThesis}.)

As mentioned in Remark \ref{remark_costello_comparison}, the effective field theory framework is often restricted to the case that $D=[Q_{GF},Q]$ is an elliptic operator, e.g.\ \cite{costello} requires $D$ to be a generalized Laplacian. This ensures that the propagator $P_{L_1,L_2}$ of Definition \ref{def:propagator} is a smooth function for all finite non-zero scales $0<L_1,L_2<\infty$. Thus, as explained in Section \ref{EFT}, the evaluation of the weight $\Phi_\gamma(iP_{L_1,L_2},\{iI_d\})(\varphi)$ of a Feynman diagram $\gamma$ involves the contraction of a distribution formed from the interaction terms $\{I_d\}$ with a smooth function formed from the propagator and the external fields. So the weight of any Feynman diagram is well-defined for finite non-zero scales $L_1, L_2$ and it is only the UV and IR limits, $L_1\to 0$ and $L_2\to \infty$, that can pose problems.

In what follows, we will be interested in applying the effective field theory framework to a gauge fixing that gives rise to a non-elliptic operator $D$. This results in a distributional propagator $P_{L_1,L_2}$, and some care must be taken in applying the Feynman rules described in Section \ref{section_perturbation_theory_details}.

To evaluate a Feynman diagram with $N$ half-edges, a distribution $u_I$ on $M^N$ given by a tensor power of interaction functionals must be contracted with a tensor product of propagators and external fields, which is now also a distribution $u_P$ on $M^N$. Such an extended notion of contraction is possible provided that the distributions satisfy a condition on their wavefront set.

\begin{defn}
\label{def_feynman_contraction}
Given a connected Feynman diagram $\gamma$ with $N$ half-edges, let ${u_I \in \D(M^N)}$, ${u_P \in \D(M^N)}$ be the distributions constructed from the interactions, respectively from the propagator and external fields, according to the Feynman rules of Section \ref{section_perturbation_theory_details}. If the wavefront sets of $u_I$ and $u_P$ satisfy
\begin{equation}
\label{WF_condition_Feynman}
    \big(\WF(u_P)+\WF(u_I)\big) \cap \{0_{M^N}\} = \emptyset,
\end{equation}
where $\{0_{M^N}\}$ is the graph of the zero-section in $T^*(M^N)$ and the sum on the left is the fiberwise sum in the cotangent bundle $T^*(M^N)$, then we define the weight of the Feynman diagram to be
\begin{equation}
\label{extended_contraction}
    \Phi_\gamma(iP_{L_1,L_2},\{iI_d\})(\varphi) = \pair{u_I}{u_P} := \pair{u_I u_P}{1}
\end{equation}
\end{defn}
Here $u_I u_P$ is the product of distributions\footnote{If $u_I$ is a distributional section of some vector bundle $F$ and $u_P$ is a distributional section of the dual vector bundle $F^*$, then the product $u_Iu_P$ in \eqref{extended_contraction} involves the fiberwise pairing in the bundles.}, which is well-defined as a distribution on $M^N$ under the wavefront set condition \eqref{WF_condition_Feynman}, see \cite[Theorem 8.2.10]{hormander1}. Note that this extended notion of contraction is the unique continuous extension of the contraction of a distribution with a smooth function, and can be viewed as a universal regularisation procedure. For the case at hand, $u_I$ and $u_P$ are given explicitly in Equation \eqref{expression_uI_uP}.

\begin{rmk}
\label{UV_IR_limits}
One could of course attempt to apply Definition \ref{def_feynman_contraction} directly to the Schwartz kernel of $Q_{GF}D^{-1}$ as distributional propagator. We choose to introduce the regularization procedure of Definition \ref{def_feynman_contraction} within the context of effective field theory, in order to regularize the Feynman diagrams for the propagator $P_{L_1,L_2}$ at finite non-zero scales. In this way the $L_1\to 0$ (UV) and $L_2\to\infty$ (IR) limits, which may still be problematic, can be dealt with separately by further regularization and by renormalization as needed. The advantage of this approach is that, whereas the wavefront set condition \eqref{WF_condition_Feynman} must hold for each diagram individually for the weight to be well-defined, the UV and IR limits are only enforced on the sum over all Feynman diagrams $\Gamma(P_{L_1,L_2},I)(\phi)$ and cancellations may occur.
\end{rmk}

\section{Ruelle Zeta Function in Perturbative Field Theory}\label{sec:RuelleZetaFT}

In this work we will be interested in interpreting certain quantities defined in the world of dynamical systems as the results of perturbative quantisation of a topological field theory called twisted, Abelian, $BF$ theory, formulated in the Batalin--Vilkovisky formalism.

The following topological setup will be used throughout. Let $M$ be an $n$-dimensional compact, connected, orientable, manifold, and let $\pi\colon E\to M$ be a rank-$r$, complex, vector bundle over $M$. We further assume that $E$ is equipped with a smooth Hermitian inner product, which we denote by $\pair{\,\cdot\,}{\cdot\,}_E$, and a flat connection $\nabla$ compatible with the inner product. Note that the flat connection induces a unitary representation of the fundamental group $\rho: \pi_1(M) \to U(\mathbb{C}^r)$. Denote by $\Omega^k(M,E)$ the space of smooth differential $k$-forms on $M$ with values in $E$. Since the connection is flat, the exterior covariant differential associated to $\nabla$,
\[ d^\nabla \colon \Omega^k(M,E) \to \Omega^{k+1}(M,E), \]
satisfies $d^\nabla\! \circ\, d^\nabla = 0$, and thus defines a cochain complex, called twisted de Rham complex, with cohomology groups $H^\bullet(M,E)$. We require this complex to be acyclic, i.e.\ $H^k(M,E) = 0$, for all $0\le k\le n$.
\begin{defn}[Twisted topological data]\label{def:TTD}
We call $(M,E,\nabla,\rho)$ the \emph{twisted topological data}.\footnote{Note that the two pieces of data $\nabla$ and $\rho$ carry essentially the same information. Nevertheless, we find it clearer to include both in the geometric package used throughout.}
\end{defn}

\subsection{Twisted, Abelian, BF Theory}
\label{BF_theory}
Given the twisted topological data of Definition \ref{def:TTD}, we can define a classical BV theory in the sense of Definition \ref{def:BV_classical}. We view $\Omega^\bullet(M,E)$ as a $\Z$-graded vector space, where $k$-forms are assigned degree $-k$. Denote by $\Omega^\bullet(M,E)[j]$ the $j$-shift of this graded vector space, that is the degree of a vector in $\Omega^\bullet(M,E)[j]$ is shifted up by $j$ compared to the degree of the corresponding vector in $\Omega^\bullet(M,E)$. Thus, a $k$-form in $\Omega^\bullet(M,E)[j]$ has degree $j-k$. The space of fields for (twisted) Abelian $BF$ theory is defined as
\begin{equation}
    \F_{BF} = \Omega^\bullet(M,E)[1] \oplus \Omega^\bullet(M,E)[n-2]
\end{equation}
Thus, a field $(\A,\B) \in \F_{BF}$ consists of a pair of inhomogeneous differential forms
$$\A = \A^{(0)} + \cdots + \A^{(n)}, \quad \B = \B^{(0)} + \cdots + \B^{(n)},$$
where $\A^{(k)}$ is an $E$-valued k-form, of total degree $|\A^{(k)}| = 1-k$, and $\B^{(k)}$ is an $E$-valued k-form of total degree $|\B^{(k)}| = n-2-k$. The symplectic pairing $\Omega_{BF}: \F_{BF}\times\F_{BF} \to \R$, is given by
\begin{equation}
    \Omega_{BF}((0,\B),(\A,0)) = \int_M \langle\B\wedge\A\rangle_E,
\end{equation}
and extended by graded anti-symmetry and linearity to all of $\F_{BF}\times\F_{BF}$. Here, $\langle\B\wedge\A\rangle_E$ denotes taking the inner product in $E$ and the exterior product in $\wedge^\bullet T^*M$, and the integral is only of the top-form part of the resulting expression. Finally, the action functional for (twisted) Abelian $BF$ theory is
\begin{equation}
\label{BF_action_traditional}
    S_{BF}\equiv S_{\d}(\A,\B) \doteq \int_M \langle\B\wedge\d\A\rangle_E.
\end{equation}

\begin{rmk}
Note that we can view $\d$ as acting on $\F_{BV}$:
$$\d (\A,\B) = (\d\A,\d\B), \quad (\A,\B) \in \Omega^\bullet(M,E)[1] \oplus \Omega^\bullet(M,E)[n-2].$$
In this way, $\d$ becomes a degree $1$ operator on $\F_{BF}$ and the action functional\footnote{In the formula below we think of $\A,\B$ as tautological (a.k.a.\ ``evaluation'') functionals.} can be written as 
$$S_{\d}(\A,\B) = \frac{1}{2}\Omega_{BF}((\A,\B), \d(\A,\B)).$$
Thanks to the compatibility of $\nabla$ with the inner product on $E$, Stokes' theorem shows that $\d$ is graded anti-symmetric with respect to the symplectic pairing.
\end{rmk}

\begin{rmk}
\label{BF_CME}
One can show that the Poisson bracket with respect to $\Omega_{BF}$ is well-defined for local functionals and that
$$\{S_{\d},S_{\d}\}_{\Omega_{BF}} = \pm 2\int_M \langle\d\B\wedge\d\A\rangle_E = \pm 2\int_M \langle\B\wedge(\d)^2\A\rangle_E = 0,$$
where $\pm$ denotes a sign depending on the degree of $\A$ and $\B$. Thus, $S_{\d}$ satisfies the classical master equation, see Definition \ref{def:BV_classical}, and $(\F_{BF},\Omega_{BF},S_{\d},\d)$ defines a classical BV theory.
\end{rmk}

We conclude this section by considering a slightly more general setting, where we replace the operator $\d$ by some operator $Q$ with similar properties as $\d$. Namely, consider a differential operator on the space of $E$-valued forms
$$Q: \Omega^\bullet(M,E) \to \Omega^{\bullet+1}(M,E),$$
satisfying $Q^2 = 0$ and graded anti-symmetry with respect to the pairing $\Omega_{BF}$. We then consider the following theory, which can be thought of as a generalisation of Abelian $BF$ theory, with the same space of fields and symplectic structure, but action functional
$$S_Q(\A,\B) = \int_M \langle\B\wedge Q\A\rangle_E.$$
This will be used in Section \ref{sec:perturbingfunctional} to view the perturbed action functional as a quadratic theory for a perturbed differential. Note that the arguments of Remark \ref{BF_CME} still apply, showing that this defines a classical BV theory.

\begin{defn}[Generalised $BF$ theory]
We call \emph{generalised BF theory} the classical BV theory specified by the data  $(\F_{BF},\Omega_{BF},S_Q,Q)$.
\end{defn}


\subsection{Ruelle zeta function}\label{sec:Ruellezeta}

Consider again the twisted topological data of Definition \ref{def:TTD}. Assume, \emph{additionally}, that $M$ is a contact manifold with contact form $\alpha$, and that it is equipped with a contact, Anosov vector field, i.e.

\begin{defn}[Contact Anosov vector fields]
\label{def_anosov}
The flow $\phi_t\colon M\to M$ of a vector field $X\in\C(M,TM)$ is Anosov if there exists a $d\phi_t$-invariant, continuous, splitting of the tangent bundle:
\begin{equation*}
    T_xM = T_0(x) \oplus T_s(x) \oplus T_u(x), \quad\text{with }\, T_0(x) = \R X(x),
\end{equation*}
i.e.\ such that $d\phi_t(T_\bullet(x)) = T_\bullet(\phi_t(x))$ with $\bullet \in \{0, s, u\}$, and such that there exist constants $C, \theta >0$ for which we have, $\forall t \ge 0$:
\begin{equation*}
    \forall v \in T_s(x):\, \|d\phi_t(x)v\| \le Ce^{-\theta t}\|v\|, \quad \forall v \in T_u(x):\, \|d\phi_{-t}(x)v\| \le Ce^{-\theta t}\|v\|.
\end{equation*}
Here $\|\cdot\|$ is the norm induced by some Riemannian metric on $M$.\footnote{The Anosov property is independent of this choice of metric, although the specific values of the constants $C, \theta$ are not.} 
We call $T_s/T_u$ the stable/unstable bundles, and we assume that they are orientable. A vector field is \emph{Anosov} iff its flow is. An Anosov vector field is \emph{contact}, on a contact manifold $M$, if there exists a contact form $\alpha$ such that $X$ is its Reeb vector field
\end{defn}

\begin{defn}[Twisted Ruelle zeta function]
Let $X$ be an Anosov vector field on $M$ with flow $\varphi_t$. The \emph{Ruelle zeta function} for the flow $\varphi$, \emph{twisted} by a (unitary) representation $\rho$ of $\pi_1(M)$, is
\begin{align}\label{ruellezetafunction}
    \zeta_\rho (\lambda) := \prod_{\gamma \in \mathcal{P}} \mathrm{det}(I - \rho([\gamma])e^{-\lambda \ell(\gamma)}),
\end{align}
where $[\gamma]\in\pi_1(M)$ is the conjugacy class of a single-winding, closed, orbit\footnote{These are often called \emph{prime} closed orbits, hence the symbol $\mathcal{P}$.} $\gamma \in \mathcal{P}$ of the flow $\phi_t$ of the Anosov vector field $X$ and $\lambda \in \mathbb{C}$.
\end{defn}

\begin{rmk}
    The (twisted) Ruelle zeta function was shown to be analytic in a half-plane of large enough $\Re(\lambda)$, and admit a meromorphic continuation to $\mathbb{C}$ \cite{butterley2007smooth, marklof2004selberg, giulietti2013anosov, DyatlovZworski_RuelleZeta,DyatlovGuillarmou}.
\end{rmk}

Let the dimension of the contact manifold be $\dim(M)=2m+1$. In the contact Anosov case the rank of the stable/unstable bundles satisfy $\mathrm{rank}(T_s)=\mathrm{rank}(T_u)=m$, see \cite{FaureTsujii_ContactAnosov}. In order to express the Ruelle zeta function in terms of certain regularized determinants, we start from the observation \cite{DyatlovZworski_RuelleZeta,hadfield_kandel_schiavina} that the (twisted) Ruelle zeta decomposes as 
\begin{align}
    \log \zeta_\rho(\lambda)
    = (-1)^m\sum_{k=0}^{2m} (-1)^k \log \zeta_{\rho,k}(\lambda).
\end{align}
An important interpretation of this comes from the Guillemin trace formula, which allows us to provide an integral representation of $\zeta_{\rho,k}(\lambda)$:
\begin{equation}\label{integralrepzeta}
    \log \zeta_{\rho, k}(\lambda) 
    = -\int_0^\infty 
    t^{-1}e^{-t\lambda} \tr e^{-t\L_{X,k}}
    \,dt.
\end{equation}
where $\L_{X,k} = \L_X\bigr|_{\Omega^k(M,E)\cap\ker(\iota_X)}$ is the Lie derivative restricted to $k$-forms lying in the kernel of $\iota_X$, and the \emph{flat} trace evaluates to \cite{guillemin1977lectures}
\begin{equation}\label{flattrace}
    \tr e^{-t\L_{X,k}}
    = \sum_{\gamma\in \mathcal{P}} \sum_{j=1}^\infty
    \ell(\gamma) \delta(t - j\ell(\gamma))
    \frac{\mathrm{tr}(\rho([\gamma])^j)\mathrm{tr} (\wedge^k P(\gamma)^j)}
    {| \det(I - P(\gamma)^j) |},
\end{equation}
as a distribution on $\R_+$\footnote{$P(\gamma)=d\phi_{-l(\gamma)}|_{T_s\oplus T_u}$ is the linearized Poincaré map associated to the closed orbit $\gamma$}. Hence, we can conclude that
\begin{thm}[\cite{DyatlovZworski_RuelleZeta,hadfield_kandel_schiavina}]\label{thm:Ruellezetadeterminant}
    The Ruelle zeta function is the flat-superdeterminant of the (degree-preserving) operator $\L_X + \lambda$ restricted to the subsapce $\ker(\iota_X) \subset \Omega^\bullet(M,E)$:
    \begin{equation}
    \label{Ruelleasadeterminant}
    \zeta_\rho (\lambda)^{(-1)^m} = \sdet\bigl((\L_{X} + \lambda)\bigr|_{\ker(\iota_X)}\bigr)  = \prod_{k=0}^{2m} {\det} (\L_{X,k} + \lambda)^{(-1)^k}.
    \end{equation}
\end{thm}

\begin{rmk}[Anisotropic Sobolev spaces]\label{rem:Anisotropic}
An important observation in this context is the relation between the poles of the meromorphic extension\footnote{The fact that the resolvent admits a meromorphic continuation is proven in \cite{faure_sjostrand,DyatlovZworski_RuelleZeta}.} of the resolvent, defined (for large $\Re(\lambda)$) as
\[
R_X(\lambda)=(\L_X+\lambda)^{-1} = \int_0^\infty e^{-t(\L_X+\lambda)}\,dt,
\] 
and the eigenvalues of the operator $\L_X$ acting on certain Hilbert spaces specially taylored to the dynamics. On such \emph{anisotropic Sobolev spaces} $H_{sG}$, $\L_X$ can be shown to have discrete spectrum. The vector field $X$ (or sometimes the Lie derivative operator $\L_X$) is said to have a Pollicott--Ruelle resonance at $\lambda$, if there are eigenvectors in some $H_{sG}$ for the eigenvalue $\lambda$. A resonance for the eigenvalue zero is then related to the resolvent having a pole at $\lambda=0$. When this is \emph{not} the case, one can ``invert'' $\L_X$ with an important caveat: it is inverted as an operator from a domain that is dense in some anisotropic Sobolev space of sufficient regularity $s$, so that $\L_X\colon \mathfrak{D}(\L_X) \subset H_{sG} \to H_{sG}$. Henceforth we will forgo the explicit construction\footnote{This is obtained from $L^2(M)$ by means of an exponential factor $e^{-sG}$ with $G$ an appropriately chosen pseudodifferential operator.} of $H_{sG}$, and refer to \cite{faure_sjostrand,DyatlovZworski_RuelleZeta} for details.
\end{rmk}

\begin{rmk}
\label{entire_function}
It can be shown that the flat determinants appearing in the alternating product \eqref{Ruelleasadeterminant} are entire functions of $\lambda \in \mathbb{C}$, whose zeros coincide with the Pollicott-Ruelle resonances in the respective form degree, see \cite{DyatlovZworski_RuelleZeta}. Thus, as befits a determinant, ${\det} (\L_{X,k} + \lambda)$ vanishes precisely when $\L_{X,k} + \lambda$ is not invertible on some anisotropic Sobolev space of sufficient regularity.
\end{rmk}

\subsection{Ruelle zeta and field theory}
\label{section_ruelle_zeta_perturbative}

In this section we apply the field theory framework, outlined in Section \ref{section_perturbation_theory}, to $BF$ Theory (see Section \ref{BF_theory}). We will employ a particular gauge fixing, introduced in \cite{hadfield_kandel_schiavina} and studied in our companion paper \cite{Schiavina_Stucker1}, which is available---given the twisted topological data of Definition \ref{def:TTD}---on contact manifolds that admit a contact Anosov flow. In this setting, the effective field theory philosophy has a nice interpretation in terms of said flow. 

In order to obtain a well-defined non-zero value for the gauge fixed free partition function, we impose the following additional condition on the Anosov flow:
\begin{assumption}\label{ass:noresonance}
The Anosov vector field $X$ has no Pollicott-Ruelle resonance at $0$, which is then in the resolvent set of the closed densely defined operator $\L_X: \mathfrak{D}(\L_X)\subset H_{sG} \to H_{sG}$, acting on the anisotropic Sobolev spaces $H_{sG}$, see \cite[Section 3.1]{DyatlovZworski_RuelleZeta}. 
\end{assumption}

\begin{defn}[Twisted, contact, Anosov data]\label{def:TCAD}
    Let $(M,E,\nabla,\rho)$ be the twisted topological data of Definition \ref{def:TTD}, and let additionally $M$ be a contact manifold of dimension $2m+1$ endowed with a contact Anosov vector field $X$ satisfying Assumption \ref{ass:noresonance}. We call the data $(M,E,\nabla,\rho,X)$ the \emph{twisted, contact, Anosov data}.
\end{defn}

Recall that in (twisted) Abelian, $BF$ theory the operator $Q_{BV}\equiv Q$ of Definition \ref{def:BV_classical} (cf.\ Remark \ref{rmk:QvsQBV}) is simply given by:
$$Q(\A,\B) = (\d\A,\d\B), \quad (\A,\B) \in \F_{BF} = \Omega^\bullet(M,E)[1] \oplus \Omega^\bullet(M,E)[n-2].$$
\begin{prop}[\cite{hadfield_kandel_schiavina,Schiavina_Stucker1}]\label{prop:GFoperator}
The operator
\[
Q_{GF} \circlearrowright \F_{BF}, \qquad Q_{GF} (\A,\B) = (\iota_X\A,\iota_X\B),
\]
is a gauge-fixing operator. In particular, defining\footnote{The notation $Q_{GF}^*$ comes from the fact that $\iota_X$ and $\alpha\wedge$ are conjugate to one another w.r.t.\ an inner product on $\F_{BF}$. See \cite{Schiavina_Stucker1} for details.}
\[
Q_{GF}^* \circlearrowright \F_{BF}, \qquad Q_{GF}^* (\A,\B) = (\alpha\wedge\A,\alpha\wedge\B)
\]
we have the Lagrangian splitting:
\[
\F_{BF}=\mathrm{im}(Q_{GF}) \oplus \mathrm{im}(Q_{GF}^*) = \mathrm{ker}(\iota_X) \oplus \mathrm{im}(\alpha\wedge),
\]
which is invariant under the action of the graded commutator:
\[D(\A,\B) = [Q_{GF},Q](\A,\B) = \L_X(\A,\B).\]
The choice of $\sfL_X := \ker(\iota_X) = \im(\iota_X)$ as gauge fixing Lagrangian will be referred to as ``the contact gauge''.
\end{prop}

\begin{rmk}
The heat kernel for $D$ is the kernel of $\phi_{-t}^*$, the pull-back by the flow of $X$. 
Thus, we can think of the propagator $P_{L_1,L_2}$ as propagating particles along the flow of $X$, from time $L_1$ to time $L_2$. The scale-$L$ effective interaction has an interpretation as the average interaction that particles feel as they flow along $X$ for a time scale $L$. In other words, the interpretation of $L$ as a length scale described in Section \ref{EFT} has given way to an interpretation of $L$ as the time scale for the flow by $X$.
\end{rmk}

As was shown in \cite{hadfield_kandel_schiavina}, we can use field theory to interpret the (absolute value of the) Ruelle zeta function, in its determinant form (Theorem \ref{thm:Ruellezetadeterminant}) in terms of the partition function of BF theory. Notice that in order to compare with Definition \ref{def:partitionfunction}, one needs to recall that $\sdet(D\vert_{\sfL_X}) = \sdet(\L_X\vert_{\im(\iota_X)})^2$, and the grading on $\F_{BF}$ differs by 1 w.r.t. the natural grading on $\Omega^\bullet(M,E)$ (see also \cite[Remarks 11 and 12]{hadfield_kandel_schiavina}).

\begin{thm}\label{thm:ruellezetazeroPF}
    Let $(M,E,\nabla,\rho,X)$ be the twisted, contact, Anosov data. Up to a phase, the partition function of BF theory in the contact gauge is the value of the Ruelle zeta function at zero:
    \begin{equation}
        Z(S_{\d},\mathsf{L}_X) = |\zeta_\rho(0)|^{(-1)^m}.
    \end{equation}
\end{thm}

\subsection{The Perturbing Functional}\label{sec:perturbingfunctional}

In this section we will consider a perturbation of the free $BF$ action by a functional $\digamma_X$, which enters at order $\h$:
$$S_{\d}(\A,\B) = \int_M \langle\B\wedge\d\A\rangle_E \leadsto S_{Q_X} \doteq S_{\d} + \hbar \digamma_X.$$

\begin{rmk}
    Assumption \ref{ass:noresonance} implies that the inverse $\L_X^{-1}: H_{sG} \to H_{sG}$ is well-defined on anisotropic Sobolev spaces. Note that both $\iota_X$ and $\d$ commute with $\L_X^{-1}$. 
\end{rmk}

\begin{defn}
Under Assumption \ref{ass:noresonance}, the perturbing functional is:
\begin{equation}
\label{interaction_term}
    \digamma_X(\A,\B) = \int_M \langle\B\wedge\L_X^{-1}\d\A\rangle_E.
\end{equation}
\end{defn}
This functional exhibits some peculiarities. The operator $W_X:= \L_{X}^{-1}\d$ appearing in the functional does not necessarily map smooth sections to smooth sections. Indeed, the condition $\L_X^{-1}(\omega)$ smooth for all $\omega \in \Omega^\bullet(M,E)$ is equivalent to the Lie derivative $\L_X: \Omega^\bullet(M,E)\to \Omega^\bullet(M,E)$ being surjective, which is not necessarily the case.
Thus, $\L_X^{-1}\d\A \in H_{sG}$ may only be a distributional section in the anisotropic Sobolev space (see Remark \ref{rem:Anisotropic}). Nevertheless, we can still integrate it against the smooth section $\B$, so the functional $\digamma_X$ is well-defined on the space of fields. However, it is not a local functional in the sense of Section \ref{section_perturbation_theory}. 

\begin{rmk}
In some sense, we can view $W_X=\L_X^{-1}\d$ as a chain contraction for the chain complex defined by our gauge fixing operator $\iota_X$. Indeed, 
$$\iota_X\L_X^{-1}\d + \L_X^{-1}\d\iota_X = \L_X^{-1}(\iota_X\d + \d\iota_X) = \L_X^{-1}\L_X = \mathrm{id}_{\mathfrak{D}(\L_X)}$$
is the identity operator on the domain of $\L_X$ in $H_{sG}$.
Since smooth sections are in $\mathfrak{D}(\L_X)$, this gives the identity on the chain complex $\Omega^\bullet(M,E)$. However, as mentioned above, $W_X$ may not actually define a map into the chain complex of \emph{smooth} sections, so $W_X$ defines a chain contraction only on appropriately defined distributional sections of the twisted cotangent bundle.
\end{rmk}

Treating $\h$ as a formal parameter, and denoting $Q_X \doteq (1 + \hbar \L_X^{-1})\d = \d +\hbar W_X$, we define our perturbed action as
\begin{equation}
\label{full_action}
S_{Q_X}(\A,\B) = S_{\d}(\A,\B) + \h \digamma_X(\A,\B) = \int_M \langle\B\wedge(1+\h\L_X^{-1})\d\A\rangle_E.
\end{equation}
Applying the Poisson bracket to $S_{Q_X}$, we find
\begin{equation*}
    \{S_{Q_X},S_{Q_X}\} = \pm 2\int_M \langle(1+\h\L_X^{-1})\d\B\wedge(1+\h\L_X^{-1})\d\A\rangle_E = 0,
\end{equation*}
where we used the fact that $\d$ and $\L_X^{-1}$ commute, so that both forms in the wedge product are in the image of $\d$. Thus, the action functional $S$ satisfies the classical master equation.

\begin{prop}
    The critical locus of $S_{Q_X}$ is given by Pollicott--Ruelle resonant states with resonance $\hbar$.
\end{prop}
\begin{proof}
The dynamical content of the perturbed classical action $S_{Q_X}=\Omega_{BV}(\B,Q_X\B)$ lies in the Euler-Lagrange equations. For the field $\A$ the Euler-Lagrange equation gives
$$Q_X\A=(1+\h\L_X^{-1})\d\A = 0.$$
Factoring out the $\L_X^{-1}$ this can equivalently be written as
$$\L_X^{-1}(\L_X+\h)\d\A = \L_X^{-1}\d(\L_X+\h)\A = 0.$$
Thus, for $\A$ to solve the equations of motion of the BV theory, we must either have
$$\d\A = 0, \quad\mathrm{or}\quad (\L_X+\h)\A = 0.$$
In other words, the nontrivial solutions are either closed forms or Pollicott-Ruelle resonant states with resonance $\h$. Restricting $\A$ to satisfy our gauge fixing condition $\A \in \im(\iota_X)$, in fact eliminates the first possibility. Indeed, if $\d\A = 0$ in this case, then also $\iota_X\d\A = \L_X\A = 0$, so $\A=0$ by the injectivity of $\L_X$. Thus, on the gauge fixing subspace the equations of motion only have nontrivial solutions if $\h$ is a Pollicott-Ruelle resonance and the solutions are given by the resonant states. The Euler-Lagrange equation for the field $\B$ can be analyzed similarly.
\end{proof}

\begin{rmk}[Dependence on $X$]
\label{remark_vector_field}
Note that the Anosov vector field $X$ enters the perturbed action functional $S_{Q_X}$ in \eqref{full_action} via the perturbing functional $\digamma_X$. In addition, we will use the vector field $X$ to define our gauge fixing condition. This is akin to what happens in Yang--Mills theory, where the action functional depends on a Riemannian metric, and the same metric is used to define the Lorenz gauge. Somewhat unusual here is that the original action functional of $BF$ theory depends solely on topological data, whereas the perturbation introduces a dependence on the additional datum $X$. As observed above, after gauge fixing the perturbed action has critical locus given by the Pollicott-Ruelle resonances with respect to $X$.
Of course, the vector field entering the gauge fixing condition is a priori independent of the vector field entering $S_{Q_X}$. One can consider a family of action functionals $S_{Q_X}$ depending on a vector field $X$, to which we can apply a family of gauge fixing conditions defined by another vector field $Y$. The most sensible result for the partition function of the theory is obtained when we set $Y=X$. An analysis of the case when the two vector fields are allowed to vary independently of each other is given in \cite{ThomasThesis}.
\end{rmk}

The perturbation $\h \digamma_X$ is actually a quadratic functional in the fields. Hence, it is worthwhile to treat the full action $S_{Q_X} = S_{\d} + \h \digamma_X$ as a free action and compute the partition function for the gauge fixing operator $\iota_X$, using the formalism of infinite dimensional Gaussian integrals employed in Definition \ref{def:partitionfunction}. 

\begin{thm}\label{thm:Result_PF}
The partition function of perturbed $BF$ theory $S_{Q_X} = S_{\d} + \hbar \digamma_X$, when computed in the contact gauge $\sfL_X = \ker(\iota_X)$ on a manifold equipped with the twisted, contact, Anosov data of Definition \ref{def:TCAD}, is the twisted Ruelle zeta function up to a phase:
\[
Z(S_{Q_X},\sfL_X) = |\zeta_X(\hbar)|^{(-1)^m}.
\]
\end{thm}

\begin{proof}

Being quadratic, we can apply the definition of partition function (Definition \ref{def:partitionfunction}) to the functional $S$, where the operator $Q$ is now $\d+\h\L_X^{-1}\d$. Thus, for the gauge fixing condition $\sfL_X = \ker(\iota_X) = \im(\iota_X)$, the evaluation of the partition function yields\footnote{The Jacobian factor for the gauge fixing operator $\iota_X$ is $1$, see \cite{hadfield_kandel_schiavina,Schiavina_Stucker1}.}
$$Z(S,\sfL_X) \doteq \big|\sdet(\iota_X(\d+\h\L_X^{-1}\d)|_{\im(\iota_X)})\big| = \big|\sdet((\L_X+\h)|_{\im(\iota_X)})\big| = |\zeta_X(\hbar)|^{(-1)^m},$$
where $\iota_X\d\vert_{\im(\iota_X)} = \L_X\vert_{\im(\iota_X)}$ cancels with the $\L_X^{-1}$ in the second term.
\end{proof}

To gain more insight on this result, we will treat the perturbation functional as an interaction term instead in what follows. This is justified by the fact that it enters at a higher order correction in $\h$ w.r.t.\ the free $BF$ action $S_{\d}$.
In the perturbative framework, we look at the expectation value of the functional $\exp(iF)$ for the gauge fixing $\iota_X$. Dividing the result above by the value of the free partition function, as in \eqref{expectation_value_finite}, this immediately gives 
\begin{cor}\label{cor:expectationvaluefromPF}
\begin{equation}
\label{F_gaussian_expectation}
    \biggl\langle e^{i\digamma_X} \biggr\rangle_{\sfL_X} \doteq \frac{Z(S_{Q_X},\sfL_X)}{Z^{\text{free}}(S_{Q_X},\sfL_X)} = \frac{Z(S_{Q_X},\sfL_X)}{Z(S_{\d},\sfL_X)} = \frac{\big|\sdet((\L_X+\h)|_{\im(\iota_X)})\big|}{\big|\sdet(\L_X|_{\im(\iota_X)})\big|} = \left|\frac{\zeta_X(\h)}{\zeta_X(0)}\right|^{(-1)^m}.
\end{equation}
\end{cor}
In the next section, we will show that the same result can be obtained from perturbation theory, and therefore without the use of Definition \ref{def:partitionfunction} for the partition function, as a generalisation of infinite-dimensional Gaussian integrals.

\subsection{Feynman diagrams for the contact gauge fixing}
\label{section_diagrams_ruelle}

In the rest of this section, we will treat the functional $\h \digamma_X$ as an interaction term for the free action $S_{\d}$, as in Section \ref{section_perturbation_theory}. This is worthwhile, since $D=\L_X$ is not elliptic, and thus a nontrivial adaptation of the perturbative quantisation setup of \cite{costello} is required.

The operators $Q, Q_{GF}$ and $D$, introduced for $BF$ theory in Section \ref{section_ruelle_zeta_perturbative}, can be represented as
\begin{equation*}
    Q = \begin{pmatrix} \d & 0 \\ 0 & \d \end{pmatrix}, \quad Q_{GF} = \begin{pmatrix} \iota_X & 0 \\ 0 & \iota_X \end{pmatrix}, \quad D = \begin{pmatrix} \L_X & 0 \\ 0 & \L_X \end{pmatrix},
\end{equation*}
with a copy of the operators acting on each component of the direct sum 
$$\F_{BF} = \Omega^\bullet(M,E)[1] \oplus \Omega^\bullet(M,E)[n-2].$$

\begin{rmk}[A comparison of gauge-fixing requirements]\label{rmk:costellospilittinganisotropic}
An operator $Q_{GF}$ is ``gauge fixing'' in the sense of \cite{costello} iff the space of fields splits as $\F_{BF} = \im(Q)\oplus\im(Q_{GF})$ (when $\ker(D)=\{0\}$). This definition is a priori different from Definition \ref{def:GFoperator}.  
Indeed, by assumption on the vector field $X$, the operator $D=\L_X$ has trivial kernel; by its injectivity, we have 
$$\im(\d)\cap\im(\iota_X) \subset \ker(\L_X) = \{0\}.$$ Furthermore, we can write a smooth section $\omega \in \Omega^\bullet(M,E)$, as
$$\omega = \L_X\L_X^{-1}\omega = \d\iota_X\L_X^{-1}\omega + \iota_X\d\L_X^{-1}\omega.$$
In this sense, $\omega$ can be written as the sum of two sections
$$\omega_1 = \d\iota_X\L_X^{-1}\omega \in \im(\d), \quad \omega_2 = \iota_X\d\L_X^{-1}\omega \in \im(\iota_X).$$
However, $\omega_1$ and $\omega_2$ are not necessarily smooth, belonging instead to anisotropic Sobolev spaces. So we cannot speak of an actual splitting of the space of \emph{smooth} forms $\Omega^\bullet(M,E)$ in this sense. Note, however, that $Q_{GF}$ is a gauge fixing operator in the sense of Definition \ref{def:GFoperator}, as shown in Proposition \ref{prop:GFoperator}. This leads to problems in applying some of the results of \cite{costello}. A Hodge-type decomposition, leading to a gauge-fixing compatible with the above considerations is explored in \cite[Section 4.2]{Dang_Riviere}. We will not give these issues any further consideration for now.
\end{rmk}

The heat kernel $K_t^{tot}$ of $D$ also consists of two components for each of the spaces making up $\F_{BF}$. Recall from Section \ref{EFT_BV} that the kernel is defined using the symplectic pairing $\Omega_{BF}$. For a section $\A\in\Omega^\bullet(M,E)[1]$, we have
$$e^{-tD}\A(x) = e^{-t\L_X}\A(x) = (K_t^{tot}*\A)(x) = \int_M K_t^{(1)}(x,y)\wedge\A(y),$$
where the integral over the $y$-coordinate is actually the pairing of a distribution with a smooth function and we suppressed the inner product in the fibers $E$ from our notation. Due to the definition of $*$ using $\Omega_{BF}$, the $y$ part of the kernel $K_t^{(1)}(x,y)$ is actually over the space $\Omega^\bullet(M,E)[n-2]$.
Similarly, for $\B\in\Omega^\bullet(M,E)[n-2]$ we have
$$e^{-tD}\B(x) = e^{-t\L_X}\B(x) = (K_t^{tot}*\B)(x) = \int_M K_t^{(2)}(x,y)\wedge\B(y),$$
Thus, the kernels $K_t^{(1)}$ and $K_t^{(2)}$ are identical as distributions, being equal to the kernel of the operator $\phi_{-t}^* = \exp(-t\L_X)$, which we just write as $K_t$. But with respect to the grading the part of $K_t^{(1)}$ acting on $k$-forms must be viewed as an element of 
$$\D(M,\wedge^kT^*M\otimes E)[1]\otimes\D(M,\wedge^{n-k}T^*M\otimes E)[n-2],$$
and similarly for $K_t^{(2)}$. The propagator of the theory is then obtained as
$$P^{tot}_{L_1,L_2} = \int_{L_1}^{L_2} (Q_{GF}\otimes1)K^{tot}_t\,dt,$$
and thus consists of two copies of
\begin{equation}
\label{propagator_scales}
    P_{L_1,L_2} = \int_{L_1}^{L_2} (\iota_X\otimes1)K_t\,dt.
\end{equation}
The two copies of this propagator, one connecting the field $\A$ to $\B$ and the other $\B$ to $\A$, will lead to a factor of two in all the Feynman diagrams we compute below.

The Lie derivative is not an elliptic operator and its heat kernel fails to be smooth. Thus, the Feynman diagrams for the contact gauge fixing must be evaluated using the extended notion of contraction in Definition \ref{def_feynman_contraction}. Our quadratic interaction functional $\digamma_X$ leads to particularly simple Feynman diagrams. At each order of $\hbar$ there are two Feynman diagrams, which are described in Section \ref{section_comp_diagrams}. In Appendix \ref{app:WFcomp} we show that theses diagrams are indeed well-defined according to Definition \ref{def_feynman_contraction}.

\begin{prop}
\label{diagrams_well_defined}
Let $P_{L_1,L_2}$ be the propagator for the contact gauge fixing and $\hbar\digamma_X$ the interaction introduced in \eqref{interaction_term}. Then for any  $0\le L_1 < L_2 < \infty$ and any connected Feynman graph $\gamma$ of the theory, we have
$(\WF(u_P)+\WF(u_I)) \cap \{0\} = \emptyset,$
and thus the weight
$$\Phi_\gamma(iP_{L_1,L_2},i\hbar\digamma_X)(\A,\B)$$
is well-defined in the sense of Definition \ref{def_feynman_contraction}.
\end{prop}

In the next section, we will compute the scale-$L$ effective interaction corresponding to $\digamma_X$:
$$i\h \digamma_X[L](\A,\B) = \Gamma(P_{0,L},\h \digamma_X)(\A,\B).$$
Note that by Proposition \ref{diagrams_well_defined} we can evaluate the Feynman diagrams directly for the propagator $P_{0,L}$. Following the philosophy outlined in Remark \ref{UV_IR_limits}, one could also first compute the diagrams for the propagator $P_{\epsilon,L}$ and then take the limit $\epsilon\to 0$. This leads to the same result, i.e.
$$\lim_{\epsilon\to 0}\Gamma(P_{\epsilon,L},\h \digamma_X)(\A,\B) = \Gamma(P_{0,L},\h \digamma_X)(\A,\B).$$
Thus, in contrast to the case where $D$ is a generalised Laplacian, our propagator $P_{\epsilon,L}$ behaves nicely as $\epsilon \to 0$, and we do not need to introduce counterterms, which could otherwise be necessary to compensate divergences at small length scales in diagrams that contain a loop. This can also be understood from the fact that, as we will see, all loop diagrams in our example involve the flat trace $\tr(\exp(-t\L_X))$, which vanishes for $t$ small enough. This is due to the fact that there exists a shortest non-trivial closed geodesic.

In the IR limit, however, our propagator is badly behaved. This is because the Lie Derivative operator $\L_X$ does not have strictly positive spectrum. Thus, we introduce a regularisation procedure, in order to regularise the propagators $P_{L_1,L_2}$ as $L_2$ goes to $\infty$. We do this by introducing a parameter $\lambda\in\mathbb{C}$ and defining
$$P_{L_1,L_2}^\lambda = \int_{L_1}^{L_2} e^{-\lambda t}(\iota_X\otimes 1)K_t\,dt .$$ 
To compute Feynman diagrams we first insert $P_{L_1,L_2}^\lambda$ for $\Re(\lambda) \gg 0$ as our propagator and then take the limit $\lambda\to 0$, applying analytic continuation in $\lambda$ as necessary. For any finite $L_2$, the limit\footnote{Certainly as a strong limit, and improvements can be made by introducing operator norms.} $\lambda\to 0$ just returns our original propagator $P_{L_1,L_2}$. It is only when setting $L_2 =\infty$ that analytic continuation becomes necessary.

\begin{rmk}
Note that the renormalization group equation \eqref{RGE} should be slightly modified in the presence of this regularization procedure. Denoting
$$i\h \digamma_X^\lambda[L](\A,\B) = \Gamma(P^\lambda_{0,L},\h \digamma_X)(\A,\B),$$
i.e. the scale-$L$ effective interaction in the presence of the regularizing parameter, the renormalization group equation should now read
$$i\h \digamma_X[L_2](\A,\B) = \lim_{\lambda\to 0}\Gamma(P^\lambda_{L_1,L_2},\h \digamma_X^\lambda[L_1])(\A,\B).$$
Once again, this distinction is only relevant when setting $L_2=\infty$.
\end{rmk}

\subsection{Computation of the scale-$L$ Effective Interaction}
\label{section_comp_diagrams}

We now compute the family of effective interactions arising from our functional $\digamma_X$. That is, we view $\h \digamma_X$ as the scale-$0$ interaction and use the perturbative expansion to compute the scale-$L$ interaction $\h \digamma_X[L]$ as in \eqref{renormalization}. Let us stress, once again, that we are treating a quadratic functional as an interaction because it is inserted at order 1 in $\hbar$.

As noted above, the propagator behaves well at small scales, and so the effective interaction at scale-$L$ can be obtained directly from the Feynman diagram expansion as
$$i\h \digamma_X[L](\A,\B) =\lim_{\lambda\to 0} \Gamma(P^\lambda_{0,L},\h \digamma_X)(\A,\B) = \Gamma(P_{0,L},\h \digamma_X)(\A,\B),$$
without the need for renormalisation (i.e.\ the addition of counterterms). Here we take $L<\infty$, so that we can ignore the regularization procedure $\lambda\to 0$.
To this end, we rewrite the interaction functional as
$$\digamma_X(\A,\B) = \int_M \langle\B\wedge W_X\A\rangle_E = \int_M \langle W_X^*\B\wedge\A\rangle_E \,,$$
where $W_X=\L_X^{-1}\d$ and $W_X^*$ is the adjoint of $W_X$ with respect to the symplectic pairing, which is just equal to $\d\L_X^{-1}$ up to a sign depending on the form degree.

The perturbative expansion 
$$\Gamma(P_{0,L},\hbar \digamma_X)(\A,\B) = \sum_{\gamma\, \mathrm{connected}} \frac{\h^{l(\gamma)}}{|\mathrm{Aut}(\gamma)|}\Phi_\gamma(iP_{0,L},i\h \digamma_X)(\A,\B)$$ is a formal power series in the parameter $\h$. In terms of Feynman diagrams, there is a factor of $\h$ attached to each vertex due to the interaction $i\h \digamma_X$, and also a factor of $\h$ for each loop in the diagram. Note that since our interaction is quadratic, we must attach exactly two propagators to each vertex. Thus, the possible connected Feynman diagrams are as follows. 

At order $\hbar$ we just get back our original interaction $i\h \digamma_X$. At order $\hbar^{N}$ for $N > 1$, there are two Feynman diagrams to consider. The first connects the external fields $\A$ and $\B$ by a chain of $N-1$ propagators, the diagram below on the left. The second is independent of the external fields and features $N-1$ propagators arranged in a loop, the diagram below on the right.

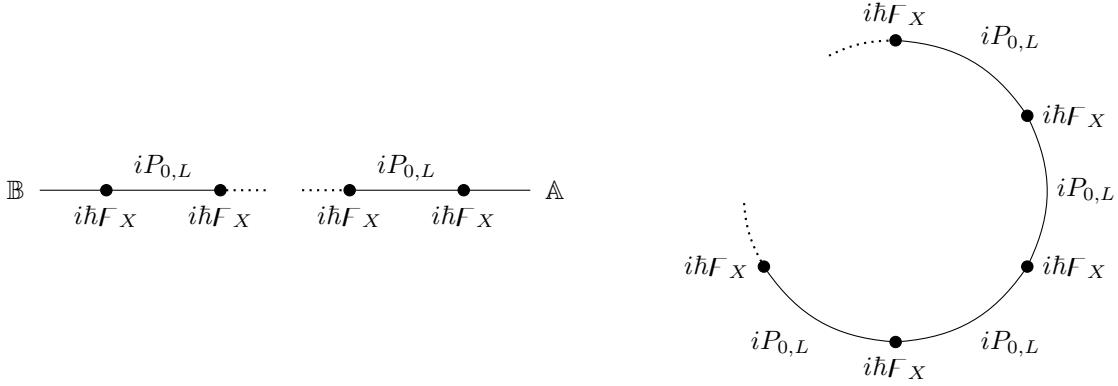
\begin{center}
\begin{tikzpicture}[baseline=(current bounding box.center)]
\begin{feynman}
\node [label=below:\(i\h \digamma_X\), dot] (a) at (0,0);
\node [label=below:\(i\h \digamma_X\), dot] (c) at (1.5,0);
\node [label=below:\(i\h \digamma_X\), dot] (d) at (3.2,0);
\node [label=below:\(i\h \digamma_X\), dot] (f) at (4.7,0);
\node (g) at (2.25,0);
\node (h) at (2.45,0);
\node [label=left:\(\B\!\!\!\)] (i) at (-1,0);
\node [label=right:\(\!\!\!\A\)] (j) at (5.7,0);
\diagram* {
(i) --
(a) -- [edge label=\(iP_{0,L}\)]
(c) -- [ghost] (g),
(h) -- [ghost]
(d) -- [edge label=\(iP_{0,L}\)]
(f) -- (j)
};
\end{feynman}
\end{tikzpicture}
\begin{tikzpicture}[baseline=(current bounding box.center)]
\begin{feynman}
\node [label=above:\(i\h \digamma_X\), dot] (a) at (0,2);
\node [label=right:\(i\h \digamma_X\), dot] (b) at (1.732,1);
\node [label=right:\(i\h \digamma_X\), dot] (c) at (1.732,-1);
\node [label=below:\(i\h \digamma_X\), dot] (d) at (0,-2);
\node [label=left:\(i\h \digamma_X\), dot] (e) at (-1.732,-1);
\node (f) at (-2,0);
\node (g) at (-1,1.732);
\node (z) at (-4,0);
\diagram* {
(g) -- [relative=true, out=12.5, in=167.5, ghost]
(a) -- [relative=true, out=25, in=155, edge label=\(iP_{0,L}\)]
(b) -- [relative=true, out=25, in=155, edge label=\(iP_{0,L}\)]
(c) -- [relative=true, out=25, in=155, edge label=\(iP_{0,L}\)]
(d) -- [relative=true, out=25, in=155, edge label=\(iP_{0,L}\)]
(e) -- [relative=true, out=12.5, in=167.5, ghost] (f)
};
\end{feynman}
\end{tikzpicture}
\end{center}

We will call the series of diagrams depending on the external fields the interaction term, $\Gamma_{int}(P_L,\h \digamma_X)$, and the series independent of the external fields the trace term, $\Gamma_{tr}(P_L,\h \digamma_X)$. In the following, we suppress the wedge product and the inner product over $E$ from our notation, since the expressions are already somewhat unwieldy as is. 
Evaluating the Feynman diagrams in the interaction term gives
\begin{equation}
\label{interaction1}
\begin{split}
    \Gamma_{int}(P_{0,L},\hbar \digamma_X)(\A,\B)\,\,& \\
    =\sum_{N=1}^\infty \hbar^N(-1)^{N-1}i &\int_{M_1}\cdots\int_{M_N}
    W_X^*\B(x_1)(1\otimes W_X^*)P_{0,L}(x_1,x_2)\cdots(1\otimes W_X^*)P_{0,L}(x_{N-1},x_N)\A(x_N) \\
    = \sum_{N=1}^\infty \hbar^N(-1)^{N-1}i &\int_0^L\,dt_1\cdots\int_0^L\,dt_{N-1} \int_{M_1}\cdots\int_{M_N} \\
    &W_X^*\B(x_1)(\iota_X\otimes W_X^*)K_{t_1}(x_1,x_2)\cdots(\iota_X\otimes W_X^*)K_{t_{N-1}}(x_{N-1},x_N)\A(x_N).
\end{split}
\end{equation}
The Feynman diagram above on the left has $N$ vertices and $N-1$ propagators, so the Feynman rules lead to a factor of $i^N\cdot i^{N-1} = (-1)^{N-1}i$. Also, the diagram has a symmetry factor of $\frac{1}{2}$, but this cancels with the factor of $2$ coming from the two copies of $P_{0,L}$ in the full propagator. We have decorated the integrals over the manifold $M$ with subscripts to denote which variable is being integrated. These integrals should actually be understood as distributional pairings.
Note that $(\iota_X\otimes W_X^*)K_t$ is precisely the kernel of the operator $\iota_X e^{-t\L_X} W_X$, since applying the adjoint $W_X^*$ to the second variable of $K_t$ corresponds to precomposing with the operator $W_X$. Thus, the integrals over the manifold in the above expression are just application of this operator, e.g.
$$\int_{M_N} (\iota_X\otimes W_X^*)K_{t_{N-1}}(x_{N-1},x_N)\A(x_N) = \big(\iota_X e^{-t_{N-1}\L_X} W_X\A \big)(x_{N-1}).$$
Applying this observation, rewriting $W_X=\L_X^{-1}\d$, and using the fact that $\iota_X$ and $\L_X$ commute, we get
\begin{equation}
\begin{split}
\label{interaction}
    \Gamma_{int}(P_{0,L},\h \digamma_X) =
    \sum_{N=1}^\infty \hbar^N(-1)^{N-1}i &\int_0^L\,dt_1\cdots\int_0^L\,dt_{N-1} \\
    &\int_M (\L_X^{-1}\d)^*\B\,\wedge\big( e^{-t_1\L_X} \L_X^{-1}\iota_X\d \cdots e^{-t_{N-1}\L_X} \L_X^{-1}\iota_X\d \A\big).
\end{split}
\end{equation}

We proceed similarly for the diagrams depicted on the right above, the trace term, and find
\begin{equation}
\label{trace1}
\begin{split}
    \Gamma_{tr}(P_{0,L},\h \digamma_X)(\A,\B)\quad& \\
    = \sum_{N=1}^\infty \h^{N+1}\frac{(-1)^N}{N} &\int_{M_1}\cdots\int_{M_N} \sum_{k=0}^n (-1)^{k+1} (1\otimes W_X^*)P_{0,L}^{(k)}(x_1,x_2)\cdots(1\otimes W_X^*)P_{0,L}^{(k)}(x_N,x_1) \\
    = \sum_{N=1}^\infty \h^{N+1}\frac{(-1)^N}{N} &\int_0^L\,dt_1\cdots\int_0^L\,dt_N \\
    &\sum_{k=0}^n (-1)^{k+1} \int_{M_1}\cdots\int_{M_N}
    (\iota_X\otimes W_X^*)K_{t_1}^{(k)}(x_1,x_2)\cdots(\iota_X\otimes W_X^*)K_{t_n}^{(k)}(x_N,x_1).
\end{split}
\end{equation}
The diagram at order $\h^{N+1}$ has $N$ vertices and $N$ propagators, with the vertices contributing a factor of $\h^N$ and an extra factor of $\h$ coming from the loop. The vertices and propagators each contribute a factor of $i^N$ giving rise to the sign $(-1)^N$. The symmetry of the diagram is that of an $N$-gon with a symmetry factor of $\frac{1}{2N}$, but again the $\frac{1}{2}$ gets canceled by the two copies of the propagator. We denoted by $P_{0,L}^{(k)}$ the part of the propagator acting on $k$-forms. The closed loop leads to a minus sign in odd degrees with respect to the grading of the fields. Recall that this grading is shifted with respect to the form degree, thus the sign of $(-1)^{k+1}$.
The integrals should again be interpreted in terms of distributional pairing. Noting once more the presence of the kernel for $\iota_X e^{-t\L_X} W_X= e^{-t\L_X}\iota_X\d\L_X^{-1}$, we recognise in the last line of equation \eqref{trace1} the flat trace of an operator, namely
\begin{equation}
\begin{split}
\label{trace}
    \Gamma_{tr}(P_{0,L},\h \digamma_X)  & \\
    =\sum_{N=1}^\infty &\h^{N+1}\frac{(-1)^N}{N} \int_0^L\,dt_1\cdots\int_0^L\,dt_N\, \\
    &\sum_{k=0}^n (-1)^{k+1}\, \tr\big(e^{-t_1\L_X} \L_X^{-1}\iota_X\d \cdots e^{-t_N\L_X} \L_X^{-1}\iota_X\d\big|_{\Omega^k(M,E)}\big)
\end{split}
\end{equation}

\begin{rmk}
In Appendix \ref{app:WFcomp}, we check the wavefront set condition, which is necessary to apply the Feynman rules with distributional propagators. For the Feynman diagrams arising from the interaction term this condition is in fact equivalent to the requirement that the composition of the operators in Equation \eqref{interaction} be well-defined. For the Feynman diagrams of the trace term the wavefront set condition in addition ensures that the flat trace in \eqref{trace} is well-defined.
\end{rmk}

\subsection{Computation of the Expectation Value}
\label{section_expectation}

The expectation value of the functional $\exp(iF)$ or, from a different perspective, the partition function of the interacting theory with the free partition function set to $1$, is computed in the perturbative framework as
$$\biggl\langle e^{i\digamma_X} \biggr\rangle_{\sfL_X} = e^{\frac{1}{\h}\Gamma(P_{0,\infty}, \h \digamma_X)(0,0)},$$
where $\sfL=\im(\iota_X)$ is the image of the gauge fixing operator entering the propagator, and we omitted the limit $\lambda\to 0$ and the regularisation $P_{0,\infty}^\lambda$, which is necessary here. The exponent on the right is the scale infinity interaction with the external fields set to $0$. As we saw in the previous section, the effective interaction consists of two terms
\begin{equation}
\label{scale_infty}
    \Gamma(P^\lambda_{0,\infty}, \h \digamma_X)(\A,\B) = \Gamma_{int}(P^\lambda_{0,\infty}, \h \digamma_X)(\A,\B) + \Gamma_{tr}(P^\lambda_{0,\infty}, \h \digamma_X).
\end{equation}
The trace term, $\Gamma_{tr}$, is a constant functional independent of the external fields $\A$ and $\B$. On the other hand, the external fields enter the interaction term, $\Gamma_{int}$, as in \eqref{interaction}. Thus, when we set $\A=\B=0$ in \eqref{scale_infty}, the interaction term vanishes. We are left with the trace term, where the scale $L$ is taken to infinity. The regularising parameter $\lambda \in \mathbb{C}$, discussed in Section \ref{section_diagrams_ruelle}, involves a factor of $e^{-t_i\lambda}$ for each of the propagators (cf. Equation \eqref{trace}). With this regularisation procedure, we have
\begin{equation}
\label{expectation1}
\begin{split}
    \Gamma(P_{0,\infty}, \h \digamma_X)(0,0) &= \Gamma_{tr}(P_{0,\infty},\h \digamma_X) \\
    &= \lim_{\lambda\to 0}\, \sum_{N=1}^\infty  \h^{N+1}\frac{(-1)^N}{N} \int_0^\infty\,dt_1\,e^{-t_1\lambda}\cdots\int_0^\infty\,dt_N\,e^{-t_N\lambda} \\
    &\hspace{7em} \sum_{k=0}^n (-1)^{k+1}\, \tr\big(e^{-t_1\L_X} (\L_X^{-1}\iota_X\d) \cdots e^{-t_N\L_X} (\L_X^{-1}\iota_X\d)\big|_{\Omega^k(M,E)}\big)
\end{split}
\end{equation}

Note the presence of the operator $\L_X^{-1}\iota_X\d$ in the above expression. Since $\iota_X\d = \L_X$ on $\im(\iota_X)$, this operator acts as the identity on $\im(\iota_X)$. In contrast, it annihilates the subspace $\im(\d)$. Thus, $\L_X^{-1}\iota_X\d$ acts like a projection onto the subspace $\im(\iota_X)$. The operator $\L_X^{-1}\iota_X\d$ is however \emph{not} a true projection operator, since it fails to map smooth functions to smooth functions, see the discussion in Section \ref{sec:perturbingfunctional}. However, inside the flat trace it plays the role of a projection operator, as the next lemma shows.

\begin{lemma}
\label{lemma_projection}
Let $B: \Omega^k(M,E) \to \D(M,\wedge^k T^*M \otimes E)$ be a continuous operator leaving the subspace $\im(\iota_X)\subset\Omega^k(M,E)$ invariant. Assume that the flat traces $\tr(B\L_X^{-1}\iota_X\d)$ and $\tr(\L_X^{-1}\iota_X\d B)$ are well-defined. Then
$$\tr(B\L_X^{-1}\iota_X\d) = \tr(B|_{\im(\iota_X)}).$$
\end{lemma}
\begin{proof}
We have a fiberwise splitting of the vector bundle $\wedge^kT^*M$ of the form
$$\wedge^kT^*M = \wedge^kT_0^*M \oplus \alpha\wedge(\wedge^{k-1}T_0^*M),$$
where $T_0^*M$ are the cotangent vectors annihilating $X$. This leads to a splitting of the smooth $E$-valued $k$-forms into
$$\Omega^k(M,E) = \Omega_0^k(M,E)\oplus\alpha\wedge\Omega_0^{k-1}(M,E).$$
The projection onto $\Omega_0^k(M,E) = \Omega^k(M,E)\cap\im(\iota_X)$ with respect to this splitting can be written as
$$\prod\nolimits_{\im(\iota_X)} = \iota_X \circ \alpha\wedge.$$
Thus,
$$\tr(B|_{\im(\iota_X)}) = \tr(B\circ\iota_X\circ\alpha\wedge).$$
Since the image of $\L_X^{-1}\iota_X\d$ is contained in $\im(\iota_X)$, we have
$$B\L_X^{-1}\iota_X\d = B(\iota_X\circ\alpha\wedge)\L_X^{-1}\iota_X\d.$$
Using the cyclicity of the flat trace, Proposition \ref{cyclicity}, this shows that
$$\tr(B\L_X^{-1}\iota_X\d) = \tr(\L_X^{-1}\iota_X\d B (\iota_X\circ\alpha\wedge)).$$
Finally, since $\L_X^{-1}\iota_X\d$ acts as the identity on $\im(\iota_X)$, and $B$ leave $\mathrm{im}(\iota_X)$ invariant, we can drop this operator from the flat trace above to get
$$\tr(B\L_X^{-1}\iota_X\d) = \tr(B (\iota_X\circ\alpha\wedge)) = \tr(B|_{\im(\iota_X)})$$
\end{proof}

Using Lemma \ref{lemma_projection}, we see that the first occurrence of $\L_X^{-1}\iota_X\d$ in \eqref{expectation1} projects down to $\im(\iota_X)$, and since this subspace is left invariant by $\exp(-t\L_X)$ the further occurrences just act as the identity and can be ignored. Taking this into account, equation \eqref{expectation1} becomes
\begin{equation}
\label{expectation}
\begin{split}
    &\Gamma(P_{0,\infty}, \h \digamma_X)(0) = \\
    &\lim_{\lambda\to 0}\, \sum_{N=1}^\infty \h^{N+1}\frac{(-1)^N}{N}\sum_{k=0}^n (-1)^{k+1} \int_0^\infty\,dt_1\cdots\int_0^\infty\,dt_N\, e^{-(t_1+\cdots +t_N)\lambda}\, \tr\big(e^{-(t_1+\cdots+t_N)\L_X}\big|_{\Omega_0^k(M,E)}\big).
\end{split}
\end{equation}
Changing variables in the integral over $t_N$ to $t = t_1+\cdots+t_N$ we find
\begin{equation}
\begin{split}
\label{t_integral}
    &\int_0^\infty\,dt_1\cdots\int_0^\infty\,dt_N\,e^{-(t_1+\cdots t_N)\lambda}\, \tr\big(e^{-(t_1+\cdots+t_N)\L_X}\big|_{\Omega_0^k(M,E)}\big) \\
    &= \int_0^\infty\,dt\,\int_0^t\,dt_1\int_0^{t-t_1}\,dt_2 \cdots \int_0^{t-t_1-\cdots-t_{N-2}}\,dt_{N-1}\,e^{-t\lambda}\,\tr\big(e^{-t\L_X}\big|_{\Omega_0^k(M,E)}\big)
\end{split}
\end{equation}
Evaluating the integrals over $t_1,\dots,t_{N-1}$ we get a factor of $t^{N-1}/(N-1)!$. So equation \eqref{t_integral} becomes
\begin{equation*}
\begin{split}
    &\int_0^\infty\,dt\,\frac{t^{N-1}}{(N-1)!}\, e^{-t\lambda}\,\tr\big(e^{-t\L_X}\big|_{\Omega_0^k(M,E)}\big) \\
    &= \frac{(-1)^{N-1}}{(N-1)!}\,\frac{d^{N-1}}{d\lambda^{N-1}}\int_0^\infty\,dt\, e^{-t\lambda}\,\tr\big(e^{-t\L_X}\big|_{\Omega_0^k(M,E)}\big).
\end{split}
\end{equation*}
For $\Re(\lambda) \gg 0$, the integral featured above converges to
$$\int_0^\infty\,dt\, e^{-t\lambda}\,\tr\big(e^{-t\L_X}\big|_{\Omega_0^k(M,E)}\big) = \frac{d}{d\lambda}\log\det\big((\L_X+\lambda)|_{\Omega_0^k(M,E)}\big) ,$$
so for $\lambda$ in this domain of the complex plane \eqref{t_integral} becomes
$$\frac{(-1)^{N-1}}{(N-1)!}\,\frac{d^N}{d\lambda^N}\log\det\big((\L_X+\lambda)|_{\Omega_0^k(M,E)}\big) .$$
Since the flat determinant is a well-defined analytic function for all $\lambda\in\mathbb{C}$, see Remark \ref{entire_function}, this expression provides an analytic continuation of the integral in \eqref{t_integral}. Inserting this expression into \eqref{expectation}, and noting that the factors of $(-1)^N$ and $(-1)^{N-1}$ combine to an overall minus sign, we find
\begin{equation*}
\begin{split}
    \Gamma(P_{0,\infty}, \h \digamma_X)(0) &= \lim_{\lambda\to 0}\,\sum_{k=0}^n (-1)^k\, \sum_{N=1}^\infty\,\frac{\hbar^{N+1}}{N!} \frac{d^N}{d\lambda^N}\log\det\big(\L_X+\lambda)|_{\Omega_0^k(M,E)}\big) \\
    &= \h\sum_{k=0}^n (-1)^k\lim_{\lambda\to 0} \big(\log\det\big((\L_X+\lambda+\h)|_{\Omega_0^k(M,E)}\big)-\log\det\big((\L_X+\lambda)|_{\Omega_0^k(M,E)}\big)\big),
\end{split}
\end{equation*}
where we recognised the Taylor expansion of $\log\det\big((\L_X+\lambda+\h)|_{\Omega_0^k(M,E)}\big)$ with respect to $\h$ up to the missing zeroth order term, which we then added and subtracted again. Since this flat determinant is an entire function on the complex plane, the Taylor expansion in $\h$ converges in a neighborhood of $\lambda$ for each $\lambda\in\mathbb{C}$. 
Taking $\lambda$ to zero and exponentiating, we finally find
\begin{equation}
\label{expectation_ruelle}
    \biggl\langle e^{i\digamma_X} \biggr\rangle_{\sfL_X} = e^{\frac{1}{\h}\Gamma(P_{0,\infty}, \h \digamma_X)(0)} = \prod_{k=0}^n\biggl(\frac{\det\big((\L_X+\h)|_{\Omega_0^k(M,E)}\big)}{\det\big(\L_X|_{\Omega_0^k(M,E)}\big)}\biggr)^{(-1)^k}
\end{equation}

In terms of the Ruelle zeta function, see Section \ref{sec:Ruellezeta}, this result can be formulated as follows.

\begin{thm}\label{thm:result_expectation}
Consider Abelian $BF$ theory for a vector bundle $E$ over a $(2m+1)$-dimensional manifold $M$ with acyclic flat connection.
Let $X$ be a contact Anosov vector field with no Pollicott-Ruelle resonance at $0$, and let $\digamma_X$ be the functional
$$\digamma_X(\A,\B) = \int_M \langle\B\wedge\L_X^{-1}\d\A\rangle_E.$$
Then the expectation value of $\exp(iF)$ with respect to the gauge fixing $\sfL_X=\im(\iota_X)$ satisfies
\begin{equation}
    \Bigl\langle e^{i\digamma_X} \Bigr\rangle_{\sfL_X} = \left(\frac{\zeta_X(\h)}{\zeta_{X}(0)}\right)^{(-1)^m},
\end{equation}
where $\zeta_X$ is the Ruelle zeta function. 
\end{thm}

\begin{rmk}[Ruelle zeta from perturbation theory]
As we had observed in Corollary \ref{cor:expectationvaluefromPF}, this expectation value can also be framed as the \emph{perturbative} evaluation of the interacting partition function, with interaction given by $\h \digamma_X$. Indeed, in the perturbation theory framework, the free partition function is generally considered to be a not-accessible quantity (and is just set to $1$), but adding the interaction term $\h \digamma_X$ leads to an interacting partition function that can be evaluated in the contact gauge as a power series in $\h$. This perturbed partition function provides information on the ratio $\zeta_X(\h) / \zeta_{X}(0)$.

However, the free partition function of $BF$ theory in the contact gauge \emph{is} accessible, when interpreted as an infinite dimensional oscillatory, Gaussian, integral as in Theorem \ref{thm:ruellezetazeroPF}. Using Theorem \ref{thm:Ruellezetadeterminant}), this provides information on the value at $0$ of the Ruelle zeta function.  Recalling that $\h\in\mathbb{C}$ is just an arbitrary complex number, we see that these two descriptions combined allow for the recovery of the full Ruelle zeta function (up to a phase) in quantum field theory.
\end{rmk}


\appendix

\section{Proof of Proposition \ref{diagrams_well_defined}}
\label{app:WFcomp}

We begin by giving an explicit description of the wavefront sets of the propagator and the interaction functional. Let $0\le L_1<L_2<\infty$ and consider the propagator
\begin{equation}
\label{propagator_eq}
    P_{L_1,L_2} = \int_{L_1}^{L_2}(\iota_X\otimes 1)K_t\,dt,
\end{equation}
where $K_t$ is the Schwartz kernel of $\exp(-t\L_X)$. Viewing $K_t$ for $t\in(L_1,L_2)$ as a distribution on $(L_1,L_2)\times M\times M$, we have by \cite[Appendix B]{DyatlovZworski_RuelleZeta}
\begin{equation}
    \WF(K_t) \subset \{(t,-\xi(X(x)),e^{tH_p}(x,\xi),x,-\xi)\,\,|\,\, t \in (L_1,L_2),\, (x,\xi)\in T^*M\},
\end{equation}
where
\begin{equation*}
    e^{tH_p}(x,\xi) = (\phi_t(x),d\phi_{-t}(\phi_t(x))^T\cdot\xi)
\end{equation*}
is the Hamiltonian flow of the principal symbol $p(x,\xi)=\xi(X(x))$. The propagator in \eqref{propagator_eq} can be written as
$$\pi_*((\iota_X\otimes 1)K_t),$$
using the push-forward of the map 
$$\pi: (t,x,y)\in (L_1,L_2)\times M\times M \to (x,y)\in M\times M.$$
See \cite[Section 6.1]{Brouder_Dang_Helein} for an introduction to the notion of push-forward of distributions. In virtue of the behavior of the wavefront set under push-forward \cite[Theorem 6.3]{Brouder_Dang_Helein}, we find that the wavefront set of the propagator can be bounded by
\begin{equation}
    \WF(P_{L_1,L_2}) \subset \{(e^{tH_p}(x,\xi),x,-\xi)\,\,|\,\, t \in (L_1,L_2),\, (x,\xi)\in T^*M,\, \xi(X(x))=0\}.
\end{equation}
Note that application of the bundle endomorphism $(\iota_X\otimes 1)$ to $K_t$ does not increase the wavefront set.

The wavefront set of the interaction functional $\digamma_X$, which is the Schwartz kernel of $W_X = \L_X^{-1}\d$, can be calculated from the explicit description of $\WF(\L_X^{-1})$ given in \cite[Proposition 3.3]{DyatlovZworski_RuelleZeta}, since composition with the differential operator $\d$ does not increase the wavefront set \cite[Theorem 8.2.14]{hormander1}. We have
\begin{equation}
\label{wavefront_resolvent}
    \WF(\digamma_X) \subset N^*\Delta \cup \Omega'_+ \cup (E_u^*\times E_s^*),
\end{equation}
where
$$N^*\Delta = \{(x,\xi,x,-\xi) \,\,|\,\, (x,\xi)\in T^*M\}$$
and
$$\Omega'_+ = \{(e^{tH_p}(x,\xi),x,-\xi)\,\,|\,\, t \in \R_+,\, (x,\xi)\in T^*M,\, \xi(X(x))=0\}.$$
$E_s^*$ and $E_u^*$ are the dual stable and unstable bundles.

The Feynman diagrams for our interaction are described in Section \ref{section_comp_diagrams}, where they are split into an interaction term and a trace term. In the notation of Definition \ref{def_feynman_contraction}, we will now compute the wavefront set of the distributions $u_P$ and $u_I$ for these diagrams, given as tensor products of propagators and interaction functionals, and check condition \eqref{WF_condition_Feynman}. 
For this one needs to know how the wavefront set behaves under tensor product, see \cite[Theorem 8.2.9]{hormander1}. For distributions $u,v$, we have
\begin{equation}
\label{tensor_WF}
    \WF(u\otimes v) \subset \WF(u)\times\WF(v) \,\cup\, \WF(u)\times(\supp(v)\times\{0\}) \,\cup\, (\supp(u)\times\{0\})\times\WF(v).
\end{equation}
We will ignore factors of $\hbar$ and $i$ in what follows.

For the diagram with $N$ vertices in the interaction term, we have
\begin{subequations}\label{expression_uI_uP}\begin{align}
u_I &= \digamma_X\otimes\cdots\otimes\digamma_X, \\
u_P &= \B\otimes P_{L_1,L_2}\otimes\cdots\otimes P_{L_1,L_2}\otimes\A,
\end{align}\end{subequations}
for smooth external fields $\A,\B$. Applying \eqref{tensor_WF} we find that
\begin{equation*}
    \WF(u_P) \subset \{(x_0,0,e^{t_1H_p}(x_1,\xi_1),x_1,-\xi_1,\dots,e^{t_{N-1}H_p}(x_{N-1},\xi_{N-1}),x_{N-1},-\xi_{N-1},x_N,0)\},
\end{equation*}
where $t_i\in(L_1,L_2)$, $x_i\in M$ and $\xi_i\in T_{x_i}^*M$ with $\xi_i(X(x))=0$. We can include the sets $\WF(u)\times(\supp(v)\times\{0\})$ and $(\supp(u)\times\{0\})\times\WF(v)$ arising from the application of \eqref{tensor_WF} by allowing $\xi_i=0$. However, note that $\xi_1=\dots=\xi_{N-1}=0$ is not included in $\WF(u_P)$.

Consider now that part of $\WF(u_I)$ arising from the $\Omega_+$ part of $\WF(\digamma_X)$. We denote this by $\mathcal{S}$. By \eqref{tensor_WF} we have
\begin{equation*}
    \mathcal{S} \subset \{(e^{s_1H_p}(y_1,\eta_1),y_1,-\eta_1,\dots,e^{s_{N}H_p}(y_{N},\eta_{N}),y_{N},-\eta_{N})\},
\end{equation*}
with $s_i>0$, $(y_i,\eta_i)\in T^*M$ and $\eta_i(X(x))=0$.

For $(\WF(u_P)+\mathcal{S})$ to intersect the zero section, there must be $s_i$, $t_i$, $(x_i,\xi_i)$, $(y_i,\eta_i)$ satisfying
\begin{equation}
\label{WF_consequence_eq}
\begin{split}
    (x_i,\xi_i) &= e^{s_{i+1}H_p}(y_{i+1},\eta_{i+1}), \quad \mathrm{for}\quad 1\le i\le N-1, \\
    (y_i,\eta_i) &= e^{t_{i}H_p}(x_{i},\xi_{i}), \quad \mathrm{for}\quad 1\le i\le N-1, \\
    y_N&=x_N, \quad \eta_N=0 \\
    x_0&=\phi_t(y_1), \quad \eta_1=0
\end{split}
\end{equation}
Together this implies that 
\begin{equation}
\label{result_zero_section}
    \xi_1=\dots=\xi_{N-1}=\eta_1=\dots=\eta_N=0,
\end{equation}
which is impossible as the wavefront sets of $u_P$ and $u_I$ do not intersect the zero section.

The set $\mathcal{S}$ came from only considering the $\Omega'_+$ part of the wavefront set for each factor of $\digamma_X$ in the tensor product $u_I$ when applying \eqref{tensor_WF}. If we replace $\Omega'_+$ with $N^*\Delta$ for the $j$-th factor in the tensor product, this changes the first equation in \eqref{WF_consequence_eq} to $(x_j,\xi_j)=(y_{j+1},\xi_{j+1})$, and the conclusion \eqref{result_zero_section} still holds. Similarly, if we replace any number of the $\Omega'_+$ by $N^*\Delta$.

Finally, we now also consider the set $E^*_u\times E^*_s$, when taking the tensor product in $u_I$. If $E_u^*\times E_s^*$ only enters in one factor, say the $j$-th factor, then the first equation in \eqref{WF_consequence_eq} is replaced by the condition 
\begin{equation}
\label{stable_unstable_condition}
    (x_{j-1},\xi_{j-1})\in E^*_u, \quad e^{t_jH_p}(x_j,\xi_j)\in E^*_s
\end{equation}
However, working from both sides, i.e. from $\eta_1=0$ and $\eta_N=0$, and using the remaining equations in \eqref{WF_consequence_eq} still leads to the conclusion \eqref{result_zero_section}. If we now allow $E^*_u\times E^*_s$ to appear in both the $j$-th and $(j+k)$-th factors of the tensor product, then \eqref{stable_unstable_condition} for $j$ and $j+k$, together with the fact that the Hamiltonian flow $\exp(tH_p)$ leaves $E^*_s$ and $E^*_u$ invariant, implies that for all $i$ with $j \le i \le j+k-1$, $\xi_i \in E^*_s\cap E^*_u = \{0\}$. This again recovers \eqref{result_zero_section}. We have now treated the full wavefront set of $u_P$, and conclude that $\WF(u_P)+\WF(u_I)$ does not intersect the zero section, so the Feynman diagrams of the interaction term are well-defined according to Definition \ref{def_feynman_contraction}.

Turning to the Feynman diagram of the trace term with $N$ vertices, we again have
$$u_I = \digamma_X\otimes\cdots\otimes\digamma_X.$$
The distribution $u_P$ is less simple to write as a tensor product, since it involves a permutation of the configuration space $M^{2N}$. Letting $\sigma$ be the diffeomorphism
$$\sigma: M^{2N}\to M^{2N}, \quad \sigma(z_1,\dots,z_{2N}) = (z_{2},\dots,z_{2N},z_1),$$
we have 
$$u_P = \sigma^*(P_{L_1,L_2}\otimes\cdots\otimes P_{L_1,L_2}).$$
The wavefront set of $u_P$ is bounded by
\begin{equation*}
    \WF(u_P) \subset \{(x_N,-\xi_N,e^{t_1H_p}(x_1,\xi_1),x_1,-\xi_1,\dots,e^{t_{N-1}H_p}(x_{N-1},\xi_{N-1}),x_{N-1},-\xi_{N-1},e^{t_NH_p}(x_N,\xi_N))\},
\end{equation*}
where $t_i\in(L_1,L_2)$, $x_i\in M$ and $\xi_i\in T_{x_i}^*M$ with $\xi_i(X(x))=0$.

Again considering first $\mathcal{S}$, the part of $\WF(u_I)$ arising from taking $\Omega'_+$ in each factor of $\digamma_X$, we find, similarly to \eqref{WF_consequence_eq}, that a point in $(\WF(u_P)+\mathcal{S})\cap\{0\}$ must satisfy
\begin{equation}
\label{WF_consequence_trace}
\begin{split}
    (x_i,\xi_i) &= e^{s_{i+1}H_p}(y_{i+1},\eta_{i+1}), \quad \mathrm{for}\quad 1\le i\le N-1, \\
    (y_i,\eta_i) &= e^{t_{i}H_p}(x_{i},\xi_{i}), \quad \mathrm{for}\quad 1\le i\le N, \\
    (x_N,\xi_N) &= e^{s_{1}H_p}(y_{1},\eta_{1})
\end{split}
\end{equation}
Writing $t=\sum_i (t_i+s_i)$ and $(x,\xi)=(x_1,\xi_1)$, these equations together imply $e^{tH_p}(x,\xi)=(x,\xi)$, i.e.
\begin{equation}
\label{fixed_point_eq}
    \phi_t(x) = x, \quad d\phi_{-t}(x)^T\cdot\xi =\xi
\end{equation}
Note that we also have $\xi(X(x))=0$, that is $\xi$ is an element of $E^*_s \oplus E^*_u$. Writing $\xi = (\xi_s,\xi_u)$, for the stable and unstable part, \eqref{fixed_point_eq} implies
\begin{equation*}
    (\xi_s,\xi_u) = (d\phi_{kt}(x)^T\cdot\xi_s,\, d\phi_{kt}(x)^T\cdot\xi_u), \quad \forall k\in \Z.
\end{equation*}
But by the Anosov property of the flow, we now have
$$\|\xi_s\| = \|d\phi_{kt}(x)^T\cdot\xi_s\| \le Ce^{-\theta kt}\|\xi_s\|, \quad \forall k\in\N,$$
for some $C,\theta >0$, and similarly for $\xi_u$ with $-k\in\N$. Since $t>0$, this forces $\xi=\xi_1$ to be $0$, and together with \eqref{WF_consequence_trace} this again leads to the conclusion \eqref{result_zero_section}.

Including the sets $N^*\Delta$ and $E_u^*\times E_s^*$ in the wavefront set of the tensor product $u_I$, can be dealt with exactly as for the interaction term discussed above. We conclude that $(\WF(u_P)+\WF(u_I))\cap\{0\}=\emptyset$, so the Feynman diagrams of the trace term are also well-defined.


\section{The Flat Determinant}
\label{section_flat_det}

In this appendix we provide a brief introduction to the flat trace and flat determinant. We follow the microlocal approach of \cite{DyatlovZworski_RuelleZeta}, see also \cite{baladi} for an approach using mollifiers. The flat determinant can be viewed as a generalisation of the more well-known zeta function regularised determinant, which is best suited for elliptic operators, whereas the flat determinant can deal with a larger class of operators.

\subsection{The Flat Trace}

Let $B:\C(M,F) \to \D(M,F)$ be a continuous operator mapping smooth sections of a vector bundle $F$ over a compact orientable manifold $M$ into distributional sections of $F$. We view the Schwartz kernel of $B$ as a distributional section\footnote{For non-orientable $M$ the top exterior product of $T^*M$ should be replaced with the density line bundle.}
\begin{equation}
\label{Schwartz_kernel}
    K_B \in \D(M\times M, F\boxtimes (F^*\otimes \wedge^n T^*M)),
\end{equation}
where $F^*$ is the dual bundle to $F$. The pull-back by the diagonal map, $\iota: x\in M \to (x,x)\in M\times M$, when it is well-defined, produces a distributional section 
$$\iota^*K_B \in \D(M,F\otimes F^*\otimes \wedge^n T^*M)$$
The flat trace is obtained from this expression by taking the fiberwise trace in $F$ and integrating over $M$. This procedure reproduces the usual trace when $B$ has smooth Schwartz kernel and is thus trace-class.

In the following we will also need a slightly extended notion of flat trace for semigroups of operators. We can view the Schwartz kernel of a strongly continuous semigroup $A_t: \C(M,F)\to \C(M,F)$ with $t\in \R_+$ as a distributional section on $\R_+\times M\times M$. We then use the extended diagonal map
$$\tilde{\iota}: (t,x) \in \R_+\times M \to (t,x,x) \in \R_+\times M\times M,$$
to pull back the Schwartz kernel, which, after integration over $M$, produces a distribution on $\R_+$.

\begin{defn}[Flat Trace]
\label{def_trace}

Let $B \colon \C(M,F) \to \D(M,F)$ be a continuous operator. Assume that its Schwartz kernel satisfies
\begin{equation}
\label{WFcondition}
\WF(K_B) \cap N^*\iota = \emptyset,
\end{equation}
where $N^*\iota$ is the set of conormals to $\iota$:
\begin{equation}
\label{conormal_diagonal}
    N^*\iota = \{(x,\xi,x,-\xi) \in T^*(M\times M) \,\mid\, (x,\xi) \in T^*M\}.
\end{equation}
Then we define the flat trace of $B$ as
\begin{equation}
\label{def_trace_formula}
    \tr(B) = \langle \mathrm{Tr}_F(\iota^*K_B), 1 \rangle.
\end{equation}

Let $A_t: \C(M,F)\to \C(M,F)$ be a strongly continuous semigroup and let $K_A$ be the Schwartz kernel of $A_t$, viewed as a distribution on $\R_+\times M\times M$. If
\begin{equation}
\label{WFcondition2}
    \WF(K_A) \cap N^*\tilde{\iota} = \emptyset,
\end{equation}
where
$$N^*\tilde{\iota} = \{(t,0,x,\xi,x,-\xi) \in T^*(M\times M) \,\mid\, t\in\R_+,\, (x,\xi) \in T^*M\}$$
we define the flat trace of $A_t$, $\tr(A_t) \in \D(\R_+)$, as
\begin{equation}
\label{def_trace_formula2}
    \pair{\tr(A_t)}{\chi} = \pair{\mathrm{Tr}_F(\tilde{\iota}^*K_A)}{\chi\otimes 1}, \quad \forall\,\chi\in\C_c(\R_+).
\end{equation}
\end{defn}

\begin{rmk}
\label{remark_schwartz_kernel}
In our applications, we will mostly consider degree preserving operators on the space of vector bundle valued forms $\Omega^\bullet(M,E)$. In this case the vector bundle $F$ in \eqref{Schwartz_kernel} is of the form $E\otimes\wedge^kT^*M$. Since $(\wedge^kT^*M)^*\otimes\wedge^nT^*M \cong \wedge^{n-k}T^*M$, we can view the Schwartz kernel as a distributional section
$$ K_B \in \D(M\times M, (E\otimes\wedge^k T^*M)\boxtimes (E^*\otimes \wedge^{n-k} T^*M)).$$
Under this identification, the action of $K_B$ on $\Omega^k(M,E)$ is given in terms of the exterior product, and the formula for the flat trace can be replaced by
$$\tr(B) = \pair{\mathrm{Tr}_E(\iota^*K_B)}{1},$$
where the pull-back already gives a top degree exterior form without the need for a fiberwise trace in $\wedge^kT^*M$
\end{rmk}

Note that the assumptions on the wavefront set of the Schwartz kernel, \eqref{WFcondition} and \eqref{WFcondition2}, are precisely the conditions for the respective pull-backs to be well-defined, see \cite[Theorem 8.2.4]{hormander1}. Identifying operators with their Schwartz kernel, \cite[Theorem 8.2.4]{hormander1} also implies:

\begin{prop}
\label{continuity}
For any closed conic set $\Gamma \subset T^*(M\times M)$ that satisfies ${\Gamma \cap N^*\iota = \emptyset}$, the flat trace is a linear functional $\tr \colon \D_\Gamma(M\times M) \to \mathbb{C}$, which is sequentially continuous with respect to the H\"ormander topology, see \cite[Definition 8.2.2]{hormander1}.

Similarly, for the case of semigroups, the flat trace is a linear map $\tr \colon \D_\Gamma(M\times M) \to \D(R_+)$, sequentially continuous with respect to the Hörmander topology for any closed conic set $\Gamma \subset T^*(\R_+\times M\times M)$ that satisfies ${\Gamma \cap N^*\tilde{\iota} = \emptyset}$.
\end{prop}

Note that by the density of $\C(M\times M)$ in $\D_\Gamma(M\times M)$, $\tr$ is in fact the unique continuous extension of the usual trace for smoothing operators to $\D_\Gamma(M\times M)$. In addition to being a continuous linear functional, the flat trace possesses another characteristic property of a trace, namely cyclicity, see \cite[Section 4.5]{VietDangDynamicalTorsion} for a proof.
\begin{prop}
\label{cyclicity}
Let $A,B: \C(M,F) \to \D(M,F)$ be two continuous operators. Assume that the compositions $A\circ B$ and $B\circ A$ are well-defined,
and both satisfy the wavefront condition in \eqref{WFcondition}. Then
$$\tr(A\circ B) = \tr(B\circ A).$$
\end{prop}

\begin{rmk}
\label{spaces}
We defined the flat trace for continuous operators from $\C(M)$ to $\D(M)$. However, if $X, Y$ are topological vector spaces with continuous inclusions $i: \C(M) \to X$ and $j: Y \to \D(M)$, then the definition automatically extends to continuous operators $B:X\to Y$. Indeed, any such $B$ induces an operator $j\circ B\circ i \colon \C(M) \to \D(M)$, whose flat trace can be computed if the wavefront condition holds. In this way, we can, for example, take the flat trace of bounded operators between $L^p$ spaces or Sobolev spaces and this flat trace will actually only depend on the restriction of the operator to smooth functions.
\end{rmk}

\subsection{The Flat Determinant}
\label{subsection_det}

We will define the flat determinant of an operator $B$ through the heat kernel of $B$, that is the semigroup $\exp(-tB)$ generated by $B$. The formula to have in mind is
\begin{equation}
\log(\det(B)) = -\frac{d}{ds}\Bigl(\frac{1}{\Gamma(s)}\int_0^\infty t^{s-1}e^{-\lambda t}\,\tr(e^{-tB}) \, dt\Bigr)\Bigr|_{s=0,\,\lambda=0},
\end{equation}
where $\Gamma(s)$ is the gamma function. This is akin to the characterization of the zeta regularized determinant via the Mellin transform. However, the above formula requires some care, as we will allow the flat trace of $\exp(-tB)$ to take values in the space of distributions, as in Definition \ref{def_trace}, and the evaluation at $s=0$, $\lambda=0$ may necessitate analytic continuation. Thus, we choose cutoff functions $\chi_N\in\C_c(\R_+)$ with $\chi_N=1$ on $[\frac{1}{N},N]$, and define the following function of two complex variables:
\begin{equation}
\label{F_B}
F_B(s,\lambda) = \lim_{N\to\infty}\pair{\tr(e^{-tB})}{\chi_N(t) t^{s-1} e^{-\lambda t}},
\end{equation}
Assuming for the moment that the limit is well-defined. If $F_B$ has an analytic continuation to $s=0$, $\lambda=0$, we define the flat determinant of $B$ as:
\begin{equation}
\label{det_F_B}
\det(B) = \exp\Bigr(-\frac{d}{ds} F_B(s,\lambda)\Big|_{s=0,\,\lambda=0}\,\Bigr).
\end{equation}
Note that letting $\lambda$ take on a non-zero value in the above formula provides a definition of $\det(B+\lambda)$.

In summary, we have the following definition.
\begin{defn}[Flat Determinant]
\label{def_det}
Let $B \colon \C(M,F) \to \D(M,F)$ satisfy the following conditions:
\begin{itemize}
    \item $B$ generates a strongly continuous semigroup of operators $e^{-tB} \colon \C(M,F) \to \C(M,F)$.
    \item The Schwartz kernel of $e^{-tB}$ viewed as a distribution on $\R_+\times M\times M$, satisfies the wavefront condition \eqref{WFcondition2}.
    \item For any smooth bounded function $f$ on $\R_+$, the sequence $\pair{\tr(e^{-tB})}{\chi_N(t)f(t) t^{s-1} e^{-\lambda t}}$ converges uniformly in $s$ and $\lambda$ as $N\to\infty$ for $\Re(s), \Re(\lambda)$ large enough.
    \item The function $F_B(s,\lambda)$ in \eqref{F_B} has an analytic continuation to $s=0,\,\lambda=0$.
\end{itemize}
Then the flat determinant of $B$ is well-defined and given by equation \eqref{det_F_B}.
\end{defn}

\begin{rmk}
Although the introduction of a bounded function in the third point above may seem ad hoc, this is just the analog of taking the absolute value, which is not available for distributions. That is, requiring convergence in the presence of a bounded function $f$, is akin to requiring absolute convergence of the integral $\int_0^\infty |\tr(e^{-tB})t^{s-1}e^{-\lambda t}|\,dt$. This ensures that the limit in \eqref{F_B} is independent of the choice of cutoff functions $\chi_N$.
\end{rmk}

\begin{rmk}
Note that we can think of $s$ as regularising the integral in $F_B(s,\lambda)$ at small $t$ and $\lambda$ as regularising the integral at large $t$. In the context of quantum field theory, where $\int_0^\infty\exp(-tB)\,dt$ has an interpretation as the propagator associated to $B$, the parameters $s$ and $\lambda$ can be thought of as regularization procedures for UV and IR divergences, respectively.
\end{rmk}

\begin{rmk}
If the operator $\exp(-tB)$ is smoothing, then the flat trace agrees with the usual trace and the flat determinant reduces to the zeta regularized determinant, see e.g.\ \cite[Proposition 6.8]{baladi}. This applies in particular when $B$ is a generalized Laplacian. One can thus think of the flat determinant as an extension of the zeta regularized determinant to operators whose heat kernel has non-trivial wavefront set.
\end{rmk}

\begin{rmk}
If $\tr(\exp(-tB))$ behaves well as $t\to 0$, so that the $N\to \infty$ limit in Definition \ref{def_det} is well-defined already for $s$ in a neighborhood of $s=0$, then the derivative at $s=0$ in equation \eqref{det_F_B} can be evaluated directly. This leads to the formula
\begin{equation*}
\log(\det(B+\lambda)) = -\int_0^\infty t^{-1}e^{-\lambda t}\,\tr(e^{-tB}) \, dt,
\end{equation*}
which should be interpreted in terms of distributional pairing and analytic continuation in $\lambda$ as in Definition \ref{def_det}. This is for instance the case when $B$ is the Lie derivative with respect to an Anosov vector field, as considered in the main part of this text.
\end{rmk}

In our applications, the function space $\C(M,F)$ has the additional structure of a graded vector space. In this context, the determinant of an operator must take into account the grading of the underlying vector space. This leads to the definition of the super determinant for degree-preserving operators. 

\begin{defn}[Flat Super Determinant]
Let $F=\bigoplus_k F^{(k)}$ be a graded vector bundle. Let $B: \C(M,F)\to\C(M,F)$ be a degree-preserving operator, with components $B^{(k)}: \C(M,F^{(k)})\to\C(M,F^{(k)})$ acting in degree $k$. Assuming that each component $B^{(k)}$ satisfies the requirements of Definition \ref{def_det}, we define the flat super determinant of $B$ as
\begin{equation}
    \sdet(B) = \prod_k \det(B^{(k)})^{(-1)^k}.
\end{equation}
Equivalently, we can obtain the flat super determinant directly form equation \eqref{det_F_B} if we replace the flat trace in \eqref{F_B} by the flat super trace:
\begin{equation}
    \str(e^{-tB}) = \sum_k (-1)^k \tr(e^{-tB^{(k)}}).
\end{equation}
\end{defn}

Finally, we require a notion of the flat determinant restricted to an invariant subspace $L$ of the operator $B$.
\begin{defn}
\label{det_restricted}
    Let $L\subset\C(M)$ be a closed subspace left invariant under the action of the semigroup $\exp(-tB)$. Furthermore, assume that there exists a continuous projection operator $\Pi_L: \C(M)\to\C(M)$ with image $L$ satisfying
    \begin{equation}
    \label{WF_projection}
        \WF(\Pi_L) \subset N^*\iota,
    \end{equation}
    see \eqref{conormal_diagonal}. Then we define the flat determinant of $B$ restricted to $L$ by replacing $\tr(\exp(-tB))$ in Definition \ref{def_det} with $\tr(\exp(-tB)\Pi_L)${, and we say that the flat determinant of $B$ \emph{restricts to} $L$.}
\end{defn}
\begin{rmk}
Note that condition \eqref{WF_projection} together with the behaviour of wavefront sets under composition of operators \cite[Theorem 8.2.14]{hormander1} imply that $\tr(\exp(-tB)\Pi_L)$ is well-defined if $\tr(\exp(-tB))$ is. Moreover, if $\Pi'_L$ is some other projection onto $L$ satisfying the requirements of Definition \ref{det_restricted}, then the cyclicity of the flat trace implies
\begin{equation*}
    \tr(\exp(-tB)\Pi_L) = \tr(\exp(-tB)\Pi'_L\Pi_L) = \tr(\Pi_L\exp(-tB)\Pi'_L) = \tr(\exp(-tB)\Pi'_L),
\end{equation*}
i.e. Definition \ref{det_restricted} is independent of the choice of projection operator.
\end{rmk}


\printbibliography

\end{document}